\documentclass[11pt]{article}

\usepackage[a4paper,left=20mm,top=20mm,right=20mm,bottom=25mm]{geometry}

\usepackage{amsmath, amsfonts, amssymb, amsthm}
\usepackage[numbers,sort&compress]{natbib}
\usepackage{authblk}
\usepackage{algorithm}
\usepackage[noend]{algpseudocode}
\usepackage{graphicx}  
\usepackage{mathtools}
\usepackage{subfig}

\usepackage[mathcal]{euscript}
\usepackage{mathptmx}

\usepackage[colorlinks]{hyperref}
\hypersetup{
	citecolor=blue,
   linkcolor=blue,
}

\usepackage[font=footnotesize,width=.85\textwidth,labelfont=bf]{caption}

\newcommand{\NEW}{}

\DeclareMathOperator{\lca}{lca}
\DeclareMathOperator{\Path}{Path}
\DeclareMathOperator{\Aho}{Aho}
\newcommand{\X}{\mathcal{X}}
\newcommand{\C}{\mathcal{C}}

\newcommand{\child}{\mathsf{child}}
\newcommand{\parent}{\mathsf{par}}
\newcommand{\IT}[1]{{r_{I}(#1)}}

\newcommand{\cl}{\ensuremath{\operatorname{cl}}}

\newcommand{\Tv}{T_{\neg v}}
\newcommand{\Tu}{T_{\neg u}}
\newcommand{\Tuv}{T_{\neg uv}}

\newcommand{\Xv}{\mathcal{X}_{\neg v}}
\newcommand{\Xu}{\mathcal{X}_{\neg u}}
\newcommand{\Xuv}{\mathcal{X}_{\neg uv}}

\newcommand{\lv}{\lambda_{\neg v}}
\newcommand{\lu}{\lambda_{\neg u}}
\newcommand{\luv}{\lambda_{\neg uv}}

\providecommand{\keywords}[1]{\textbf{\textit{Keywords: }} #1}

\newtheorem{theorem}{Theorem}
\newtheorem{lemma}{Lemma}
\newtheorem{corollary}{Corollary}
\newtheorem{problem}{Problem}

\newtheorem{definition}{Definition}

\begin{document}

\title{Reconstructing Gene Trees From Fitch's Xenology Relation}

\author[1]{Manuela Gei{\ss}}
\author[1]{John Anders}
\author[1,5,6,7]{Peter F.\ Stadler}
\author[8]{Nicolas Wieseke}
\author[2,3]{Marc Hellmuth}

\affil[1]{Bioinformatics Group, Department of Computer Science; and
		    Interdisciplinary Center of Bioinformatics, University of Leipzig, \\
			 H{\"a}rtelstra{\ss}e 16-18, D-04107 Leipzig}
\affil[2]{\footnotesize Dpt.\ of Mathematics and Computer Science, University of Greifswald, Walther-
  Rathenau-Strasse 47, D-17487 Greifswald, Germany \\
	\texttt{mhellmuth@mailbox.org} }
\affil[3]{Saarland University, Center for Bioinformatics, Building E 2.1, P.O.\ Box 151150, D-66041 Saarbr{\"u}cken, Germany }
\affil[4]{Department of Mathematics and Computer Science,
		    University of Southern Denmark, Denmark }
\affil[5]{Max-Planck-Institute for Mathematics in the Sciences, \\
  Inselstra{\ss}e 22, D-04103 Leipzig}
\affil[6]{Inst.\ f.\ Theoretical Chemistry, University of Vienna, \\
  W{\"a}hringerstra{\ss}e 17, A-1090 Wien, Austria}
\affil[7]{Santa Fe Institute, 1399 Hyde Park Rd., Santa Fe, USA} 
\affil[8]{Parallel Computing and Complex Systems Group \\
  Department of Computer Science, 
  Leipzig University \\
  Augustusplatz 10, 04109, Leipzig, Germany}
\date{}
\normalsize

\maketitle

\abstract{
  Two genes are xenologs in the sense of Fitch if they are separated by at least one horizontal gene transfer event. Horizonal gene transfer is asymmetric in the sense that the transferred copy is distinguished from the one that remains within the ancestral lineage. Hence xenology is more precisely thought of as a non-symmetric relation: $y$ is xenologous to $x$ if $y$ has been horizontally transferred at least once since it diverged from the least common ancestor of $x$ and $y$. We show that xenology relations are characterized by a small set of forbidden induced subgraphs on three vertices. Furthermore, each xenology relation can be derived from a unique least-resolved edge-labeled phylogenetic tree. We provide a linear-time algorithm for the recognition of xenology relations and for the construction of its least-resolved edge-labeled phylogenetic tree. The fact that being a xenology relation is a heritable graph property, finally has far-reaching consequences on approximation problems associated with xenology relations. }

\bigskip
\noindent
\keywords{Fitch Xenology;
          Phylogenetic Tree;
           Least-Resolved Tree;
          Informative Triple Sets;
           Di-Cograph; 
 					Heritable Graph Property;
					Forbidden Induced Subgraphs;
    Recognition Algorithm;
		Fixed Parameter Tractable
}

\sloppy

\section{Introduction}

The history of a gene family is defined by a phylogenetic tree of the genes
(the \emph{gene tree}), together with an \emph{event labeling} of its inner
vertices that identifies gene duplications, speciations events, and
possibly horizontal gene transfer, as well as a mapping of the gene tree
onto a \emph{species tree}. The latter provides an implicit dating of the
events that generated the gene phylogeny relative to the phylogeny of the
species under consideration. The mathematical structure of gene family
histories, i.e., the mutual relationships between gene trees, event
labelings, species trees, and the corresponding reconciliation maps has
only recently been explored in detail.

The concept of \emph{orthologs}, that is, pairs of genes from different
species that arose from a speciation event \cite{Fitch:70}, play a key role
in evolutionary biology. While functional similarity is not a defining
feature of orthology, in general, orthologous genes from
closely related species have a similar
function. More strictly, one-to-one orthologs are in most cases
functionally equivalent. Paralogs, that is pairs of genes that arose from
duplication events, in contrast, often have related, but clearly distinct
functions \cite{Koonin:05}. Orthologs, furthermore, tend to evolve in a
clock-like fashion (at least as long as there are no additional duplications),
which makes them the characters of choice in molecular phylogenetics
\cite{gabaldon2013}.

The orthology relation on a set of genes forms a co-graph, whose associated
co-tree is a not necessarily fully resolved event-labeled gene tree
\cite{Boecker:98,Hellmuth:13a}. This result in particular implies that
empirical estimates of the orthology relation, which are feasible in
practise \cite{Altenhoff:16}, provide direct information on the gene
history. Empirically estimated orthology relationships in general violate
the co-graph property, suggesting co-graph editing as a means to correct
the initial estimate
\cite{Lafond:13,Lafond:14,Lafond:16,DEML:16,lafond2015orthology,DONDI:17}. The
event-labeled gene trees in turn constrain the possible species trees with
which they can be reconciled \cite{HernandezRosales:12a,HW:16b}. Given data
on enough gene families, these constraints can be strong enough to
completly specify also the species phylogeny
\cite{hellmuth_phylogenomics_2015}.

Horizontal gene transfer is intimately related to the concept of
\emph{xenology}. A formal definition of xenology is less well established
and by no means consistent in the biological literature. First we note that
horizontal transfer is intrinsically a directional event, i.e., there is a
clear distinction between the horizontally transferred ``copy'' and the
``original'' that continues to be vertically transferred. This fact can be
annotated in the gene tree by associating a label to the edge that points
from the horizontal transfer event to the next event in the history of the
copy \cite{Hellmuth:17,NGD+17}.  In
\cite{Hellmuth:16a,hellmuth_partialHomology} this label was interpreted as
a direction, leading to a notion of directed co-graphs \cite{Crespelle:06},
which turned out to be intimately related to so-called uniformly non-prime
2-structures, see \cite{Hellmuth:16a}.

The most commonly used definition in the biological literature, introduced
by Walter M.\ Fitch in 2000, calls a pair of genes \emph{xenologs} if the
history since their common ancestor involves horizontal transfer of at
least one of them \cite{Fitch:00,Jensen:01}.  Preserving the directionality
of horizontal transfer, we capture this concept with the help of a
non-symmetric xenology relation $\mathcal{X}$ on a set of genes such that
$(x,y)\in\mathcal{X}$ whenever there is at least one directed horizontal
transfer event during the evolution from the last common ancestor of $x$
and $y$ towards $y$.

\NEW{While best match heuristics have been very successful as
  approximations of the orthology relation \cite{Altenhoff:16,Nichio:17},
  no comparable approach to extract the xenology relation directly from
  (dis)similarity data has been devised to-date. We suspect that this is at
  least one reason why the binary xenology relation has attracted very
  little attention so far.  Nevertheless, there are several methods to
  detect xenologs in a genome that use sequence features rather then
  phylogenetic reconstructions, see e.g.\ \cite{Ravenhall:15,Rancurel:17}.
  In this contribution we focus on the mathematical properties of the
  xenology relation $\mathcal{X}$. In particular, we will be concerned with
  two related questions: (1) How much information on the gene tree $T$ and
  the location of the horizontal transfer events within $T$ is contained in
  the xenology relation? (2) Is it possible to extract the topological
  information and labeling information from $\mathcal{X}$ efficiently?}

\NEW{We show here} that valid xenology relations correspond to a heritable
family of di-graphs, which we call the Fitch graphs.  \NEW{These} are
characterized by a small set of forbidden subgraphs on three vertices and
thus can be recognized in cubic time. Fitch graphs form a subclass of
di-cographs, which have recently been associated with an alternative
concept of xenology \cite{Hellmuth:16a}. Each Fitch graph is explained by a
unique least-resolved edge-labeled phylogenetic tree. \NEW{This tree is
  displayed by the full evolutionary scenario. It therefore provides a
  least partial information on the gene tree and the placement of the
  horizontal transfer events. We will show, furthermore, that this tree as
  well as corresponding the edge-labeling can be constructed from
  $\mathcal{X}$} in polynomial time. Utilizing \NEW{features} of heritable
graph properties we derive a linear-time recognition algorithm, as well as
NP-completeness and fixed-parameter tractable results for the respective
graph modification problems. \NEW{We take these results as motivation for
  future work towards methods to estimate the xenology relation from
  sequence (dis)similarity data.}

\section{Preliminaries: Rooted Trees, Phylogenetic Trees and Rooted  Triples}

A \emph{rooted tree} $T=(V,E)$ with {\em leaf set} $L\subseteq V$ (or $L(T)$ in
case of ambiguity) and {\em inner} vertices $V^0=V\setminus L$ is an acyclic
connected graph containing one distinguished inner vertex $\rho_T\in V^0$
that is called the \emph{root of T}. The \emph{degree} of a vertex $v \in V$ 
is denoted by $\deg(v)$.
The root $\rho_T$ is regarded as an
inner vertex, i.e.\ $\rho_T\notin L$, even if $\deg(\rho_T)=1$.  A rooted
tree $T=(V,E)$ on $L$ is \emph{phylogenetic} if its root has at least
$\deg(\rho_T) \ge 2$ and every other inner vertex $v\in V^0\setminus\{\rho_T\}$ has
$\deg(v) \ge 3$.  If the degree of each vertex $v\in
V^0\setminus\{\rho_T\}$ is exactly three and $\deg(\rho_T)=2$, then the
phylogenetic tree is called \emph{binary}. In this contribution, we will
consider rooted trees together with an edge-labeling $\lambda:E\to\{0,1\}$
and write $(T,\lambda)$.

We call $u\in V$ an \emph{ancestor} of $v\in V$, $u\succeq_T v$, and $v$ a
\emph{descendant} of $u$, $v\preceq_T u$, if $u$ lies on the unique path
from $\rho_T$ to $v$. We write $v\prec_T u$ ($u\succ_T v$) for $v\preceq_T
u$ ($u\succeq_T v$) and $u\neq v$. If $v\preceq_T u$ or $u\succeq_T v$,
then $u$ and $v$ are \emph{comparable}, and \emph{incomparable} otherwise.
It will be convenient to use a notation for edges $e$ that implies which of
the vertex in $e$ is closer to the root, that is, we always write $(u,v)\in
E$ to indicate that $u\succ_T v$. In the latter case, vertex $u$ is called
\emph{parent} of $v$, denoted by $\parent(v)$. Similarly, we define
\emph{the children} of $u$ as $\child(u):=\{v\in V \mid (u,v)\in E\}$.  We
denote two leaves $v,w \in L$ as \emph{siblings} if $v,w\in
\child(u)$. Edges that are incident to a leaf are called \emph{outer
  edges}. Conversely, \emph{inner edges} do only contain inner vertices.

For a non-empty subset $L'\subseteq L$ of leaves, \emph{the least common
  ancestor of} $L'$, denoted as $\lca_T(L')$, is the unique
$\preceq_T$-minimal vertex of $T$ that is an ancestor of every vertex in
$L'$. We will make use of the simplified notation
$\lca_T(x,y,z):=\lca_T(\{x,y,z\})$ for $L'=\{x,y,z\}$ and we will omit the
explicit reference to $T$ whenever it is clear which tree is
considered. Analogously, we often write $\deg(v)$ instead of $\deg_T(v)$
for the degree of some vertex $v$.

A \emph{simple contraction} of an edge $e=(x,y)$ in a tree $T$ refers to
the removal of $e$ and identification of $x$ and $y$.  The tree $T(L')$
with root $\lca_T(L')$ has leaf set $L'$ and consists of all paths in $T$
that connect the leaves in $L'$.  We say that a rooted tree $T$ on $L$
\emph{displays} a root tree $T'$ on $L'$, in symbols $T'\le T$, if $T'$ can
be obtained from $T(L')$ by a sequence of simple edge contractions.  We
write $T'<T$ if $T'\le T$ and $T'\ne T$. The \emph{restriction} $T|L'$ of
$T$ to $L'$ is the rooted tree obtained from $T(L')$ by suppressing all
vertices of degree $2$ with the exception of the root $\rho_T$ if
$\rho_T\in V(T(L'))$.  By construction, $T|L'$ is a phylogenetic tree. The
suppression of vertices of degree $2$ can be achieved by simple contraction
of one of the adjacent edges. Moreover, $T|L'\le T$, i.e., $T$ displays the
restrictions $T|L'$ to all subsets $L'\subseteq L$. Note that $T|L=T$ if
and only if $T$ is phylogenetic; otherwise $T|L < T$.

For every vertex $v\in V(T)$ we denote by $C(v)$ the subset of $L$ such
that $\forall x \in L$ it holds that $x \in C(v)$ if and only if $x \preceq_T v$.
Moreover, we define $\mathcal{C}(T) \coloneqq \{C(v)\mid v \in V(T)\}$.
A rooted tree is phylogenetic if and only if
$C(u)=C(v)$ implies $u=v$ for all $u,v\in E(T)$. We say that a rooted tree
$T'$ on $L$ \emph{refines} a rooted tree $T$ on $L$, if $T'$ displays
$T$. In particular, a phylogenetic tree $T'$ on $L$ refines a rooted tree
$T$ if and only if $\mathcal{C}(T)\subseteq \mathcal{C}(T')$. In
particular, the tree $T(v)$ rooted at a vertex $v$ of $T$ is the tree
$T(C(v))$. 

\emph{Rooted triples} are binary rooted phylogenetic trees on three leaves.
We write $ab|c$ for the rooted triple with leaves $a,b$ and $c$, if the
path from its root to $c$ does not intersect the path from $a$ to $b$.  The
definition of ``display'' implies that a triple $ab|c$ with $a,b,c\in L$ is
\emph{displayed} by a rooted tree $T$ if $\lca(a,b)\prec_T \lca(a,b,c)$.

The set of all triples that are displayed by $T$ is denoted by $r(T)$.  For
a set $R$ of rooted triples we define $R_x\subseteq R$ as the set of
triples in $R$ that contain the leaf $x$. A set of rooted triples $R$ is
called \emph{consistent} if there exists a phylogenetic tree $T$ on
$L_R\coloneqq\bigcup_{ab|c\in R} \{a,b,c\}$ that displays $R$, i.e.,
$R\subseteq r(T)$. In particular, a tree can display at most one triple on
any set of three leaves. Thus a triple set $R$ is inconsistent whenever
$ab|c, ac|b \in R$. However, triple sets can be inconsistent even if they
do not contain two triples on the same three leaves.

Rooted triples are widely used in the context of supertree reconstruction
because every phylogenetic tree $T$ is identified by its triple set $r(T)$,
and $r(T)\subseteq r(T')$ if and only if $T'$ displays $T$
\cite{semple_phylogenetics_2003}.  As a consequence, supertree
reconstruction can be phrased in terms of triples. As shown in
\cite{aho_inferring_1981} there is a polynomial-time algorithm, usually
referred to as \texttt{BUILD}
\cite{semple_phylogenetics_2003,steel_phylogeny:_2016}, that takes a set
$R$ of triples as input and either returns a particular phylogenetic tree
$\Aho(R)$ that displays $R$, or recognizes $R$ as inconsistent.

The requirement that a set $R$ of triples is consistent, and thus, that
there is a tree displaying all triples, makes it possible to infer new
triples from the trees that display $R$ and to define a \emph{closure
  operation} for $R$
\cite{grunewald_closure_2007,bryant_extension_1995,HS:17,Bryant97}.  Let
$\langle R \rangle$ be the set of all rooted trees with leaf set $L_R$ that
display $R$.  The closure of a consistent set of rooted triples $R$ is
defined as
\begin{equation*}
  \cl(R) = \bigcap_{T\in \langle R\rangle} r(T).
\end{equation*}
Hence, a triple $r$ is contained in the closure $\cl(R)$ if all trees that
display $R$ also display $r$.  This operation satisfies the usual three
properties of a closure operator \cite{bryant_extension_1995}, namely: (i)
expansiveness, $R \subseteq \cl(R)$; (ii) isotony, $R' \subseteq R$ implies
that $\cl(R')\subseteq \cl(R)$; and (iii) idempotency,
$\cl(\cl(R))=\cl(R)$. Since $T \in \langle r(T)\rangle$, it is easy to see
that $\cl(r(T))=r(T)$ and thus, $r(T)$ is always closed.

A set of rooted triples $R$ \emph{identifies} a tree $T$ with leaf set
$L_R$ if $R$ is displayed by $T$ and every other tree $T'$ that displays
$R$ is a refinement of $T$.  A rooted triple $ab|c\in r(T)$
\emph{distinguishes} an edge $(u,v)$ in $T$ iff $a$, $b$, and $c$ are
descendants of $u$, $v$ is an ancestor of $a$ and $b$ but not of $c$, and
there is no descendant $v'$ of $v$ for which $a$ and $b$ are both
descendants. In other words, $ab|c\in r(T)$ distinguishes the edge $(u,v)$
if $\lca(a,b)=v$ and $\lca(a,b,c)=u$.

We will make use of two results from \cite{grunewald_closure_2007} that are
closely related to the \texttt{BUILD} algorithm.
\begin{lemma}\label{g1}
  Let $T$ be a phylogenetic tree and let $R$ be a set of rooted triples.
  Then, $R$ identifies $T$ if and only if $\cl(R)=r(T)$.  Moreover, if $R$
  identifies $T$, then $\Aho(R)=T$.
\end{lemma}

\section{The (Fitch-)Xenology Relation}

\label{sec:fitch-graph}

In this contribution we are specifically interested in phylogenetic trees
$T=(V,E)$ with leaf set $L=L(T)$ that are endowed with edge labels
$\lambda:E\to\{0,1\}$ such that
\begin{equation*}
\lambda(e) = 
\begin{cases}
  1 & \text{if $e$ is a horizontal transfer-edge}\\
  0 & \text{otherwise}
\end{cases}
\end{equation*}
For simplicity we will speak of 0-edges and 1-edges in $T$ depending on
their labeling.

\begin{definition}
  Given an edge-labeled phylogenetic tree $(T,\lambda)$ we set
  $(x,y)\in\X_{(T,\lambda)}$ for $x,y\in L$ whenever there is at least one
  directed horizontal transfer event between $y$ and the last common
  ancestor of $x$ and $y$, i.e., if the uniquely defined path from
  $\lca_T(x,y)$ to $y$ contains at least one 1-edge. We write $[x,y]\in\X_{(T,\lambda)}$
  if $(x,y)$ and $(y,x)\in\X_{(T,\lambda)}$ and $x|y$ if $(x,y)$ and $(y,x)\notin\X_{(T,\lambda)}$.
\end{definition}

By construction $\X_{(T,\lambda)}$ is irreflexive; hence it can be regarded
as a simple directed graph. In the following, we therefore will
interchangeably speak of $\X_{(T,\lambda)}$ as graph or relation and use
the standard graph terminology such as ``induced subgraph in
$\X_{(T,\lambda)}$''.  It is easy to check that $\X_{(T,\lambda)}$ is in
general neither symmetric nor antisymmetric. The relation
$\X_{(T,\lambda)}$ formalizes Fitch's concept of \emph{xenology}
\cite{Fitch:00}.

We say that an edge-labeled phylogenetic tree $(T,\lambda)$ \emph{explains}
a given irreflexive relation $\X$ whenever $\X=\X_{(T,\lambda)}$. To be
more explicit, $(T,\lambda)$ \emph{explains} $\X$ if there is a 1-edge on
the path from $\lca(x,y)$ to $y$ if and only if $(x,y)\in \X$. By
construction, $\X$ must be defined on $L(T)$.  We call a relation $\X$
\emph{valid} if there is an edge-labeled tree that explains $\X$. \NEW{An
  example of a gene tree with the corresponding Fitch relation $\X$ and an
  edge-labeled tree that explains $\X$, can be found in Fig.\
  \ref{fig:GeneTree}.}

\begin{figure}[t]
\begin{center}
  \includegraphics[width=0.8\textwidth]{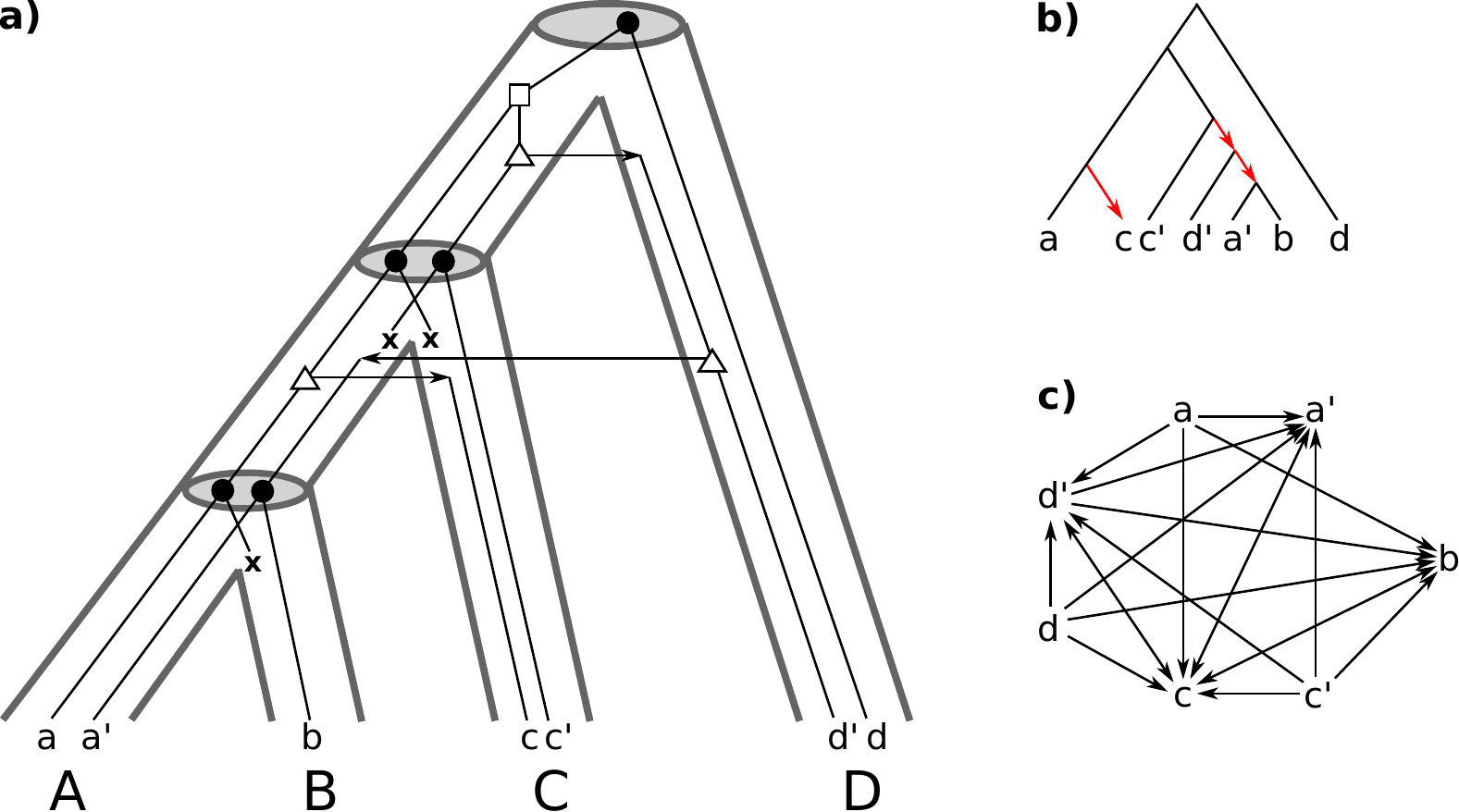}
\end{center}
\caption{\NEW{a) Event-labeled gene tree embedded in the (tube-like)
    species tree. The leaf set of the gene tree are the genes $a$, $a'$,
    $b$, $c$, $c'$, $d$, and $d'$ in the genomes of the four species $A$,
    $B$, $C$ and $D$.  The gene tree contains speciations ($\bullet$),
    duplications ($\square$), HGT events ($\triangle$) and gene losses
    ($\times$). b) Removal of all gene losses, suppression of all resulting
    degree two vertices and ignoring the types of the events on the
    vertices yields an edge-labeled tree in which the transfer edges
    labeled by $1$ (red arrow) and all other edges by $0$ (black
    edges). Panel c) shows the Fitch graph explained by the edge-labeled
    tree of Panel b).}  
}
\label{fig:GeneTree}
\end{figure}

\begin{figure}[t]
\begin{center}
  \includegraphics[width=\textwidth]{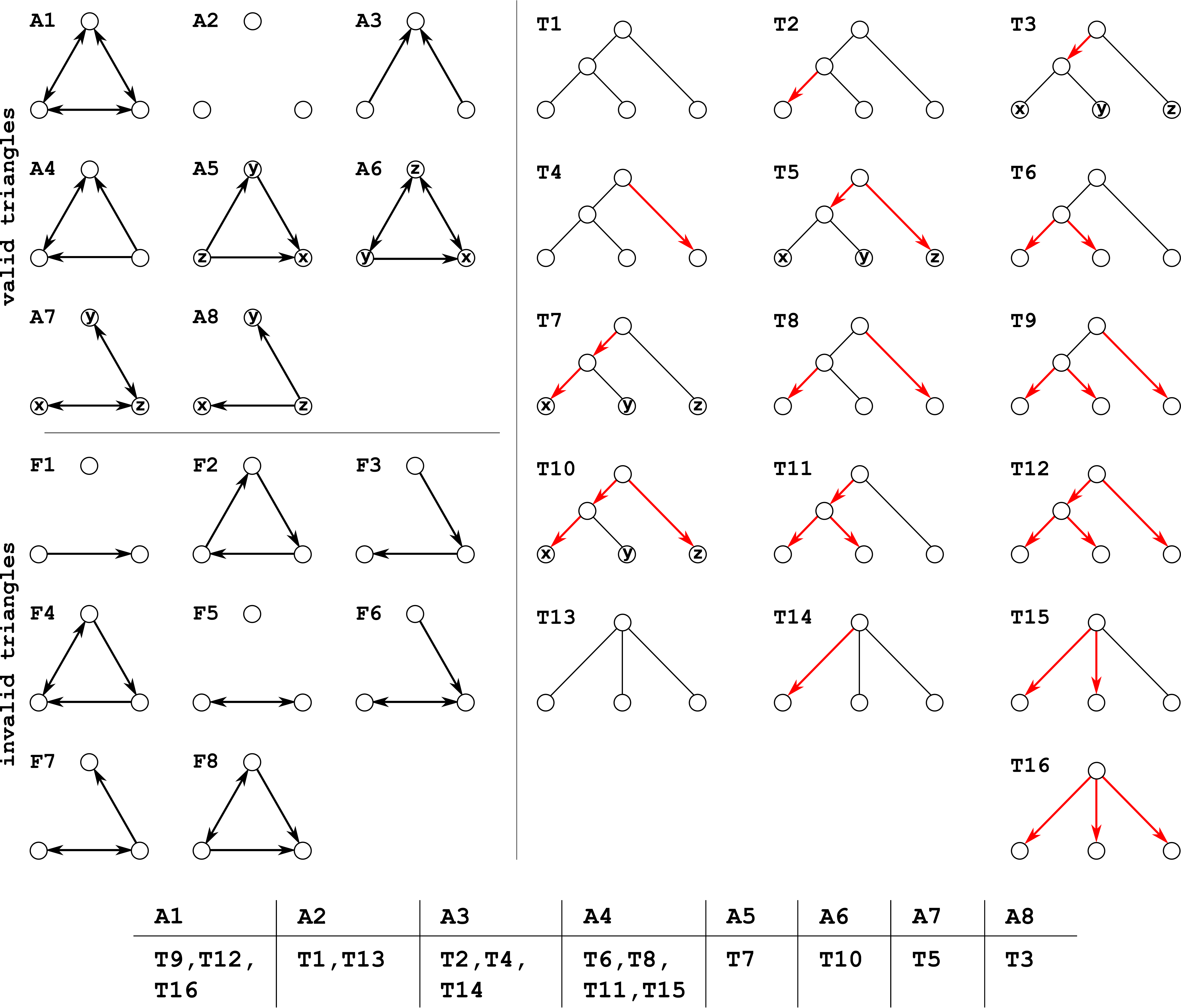}
\end{center}
\caption{\emph{Upper Left:} Shown is the graph representation for all
  possible relations $\X\subseteq L\times L$ with $|L|=3$.  The relations
  are grouped into valid ($A_1-A_8$) and non-valid ($F_1-F_8$). \\
  \emph{Upper Right:} All possible (up to isomorphism) subtrees on three
  leaves of a tree $(T,\lambda)$ are shown.  Edges can be understood as
  paths, whereby red (resp. black) edges indicate that there is (resp., is
  not) a  1-edge on the particular path. \\
  \emph{Lower Part:} The table shows which tree explains which
  relation. In particular, there is no tree that would explain one of the
  graphs $F_1$ to $F_8$.}
\label{fig:triangles}
\end{figure}

The notion of a tree $T'$ being displayed by a tree $T$ can be generalized to
edge-labeled trees: We say that $(T',\lambda')$ is displayed by
$(T,\lambda)$ if $T'$ is displayed by $T$ in the usual sense and an edge
$e'\in E(T')$ has label $\lambda'(e')=1$ if and only if the path in $T$
that corresponds to $e'$ contains at least one 1-edge.

\begin{lemma}
\label{lem:induced}
  Let $(T',\lambda')$ be a tree with leaf set $L'=L(T')$ that is displayed
  by $(T,\lambda)$. Then $\X_{(T',\lambda')}$ is the subgraph of 
  $\X_{(T,\lambda)}$ induced by $L'$. 
\end{lemma}
\begin{proof}
  Consider two distinct leaves $x,y\in L'$. By construction of
  $(T',\lambda')$ there is a 1-edge on the path from $\lca_{T'}(x,y)$ to
  the leaf $y$ in $(T',\lambda')$ if and only if the corresponding path in
  $(T,\lambda)$ containes a 1-edge and thus $(x,y)\in\X_{(T',\lambda')}$
  iff $(x,y)\in\X_{(T,\lambda)}$.
 \end{proof}

The enumeration of all edge-labeled trees on two vertices shows that all
four possible digraphs on two vertices are valid. For three vertices, however,
there are valid and invalid digraphs. These are summarized in Figure
\ref{fig:triangles}: up to isomorphism there are eight valid $A_1$-$A_8$
and eight invalid $F_1$-$F_8$ digraphs. We will refer to them as valid and
invalid \emph{triangles}. We denote subgraphs of $\X$ that are induced by
the vertices $x_1,\dots,x_k$ by $\X[x_1,\dots,x_k]$.  In particular,
triangles in $\X$ are denoted by $\X[a,b,c]$, where $a,b,c\in L$ are three
distinct vertices.

\begin{definition} 
  An irreflexive binary relation $\X$ on $L$ is a \emph{Fitch relation} if
  all its triangles are valid. Its graph representation is called a \emph{Fitch
    graph.}
\end{definition} 

\begin{figure}[t]
\begin{center}
  \includegraphics[width=0.8\textwidth]{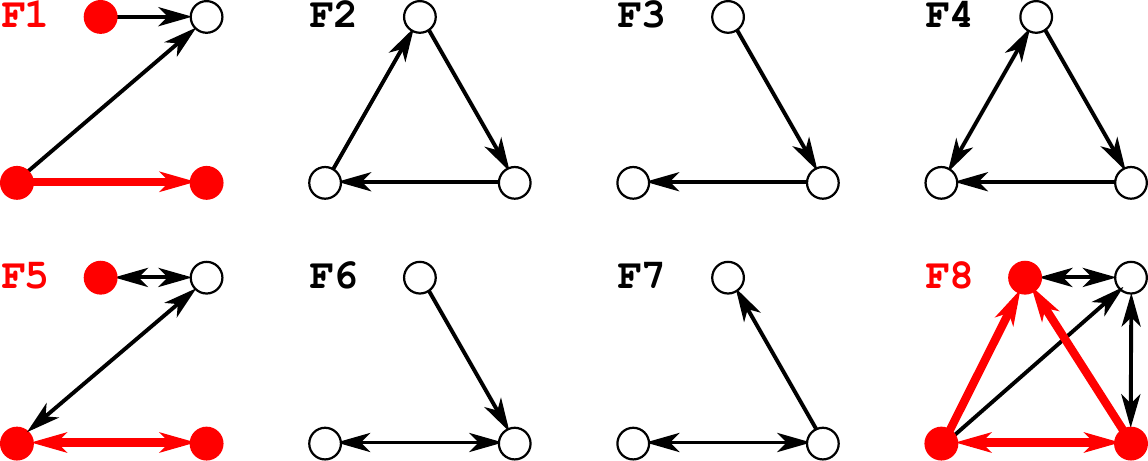}
\end{center}
\caption{The eight digraphs are the forbidden induced subgraphs that
  characterize di-cographs \cite{Ehrenfeucht:90,Crespelle:06}. The five
  digraphs on three vertices correspond to five of the eight forbidden
  triangles. Each digraph on four vertices contains one of the remaining
  forbidden triangles (highlighted by bold-red edges and vertices).}
\label{fig:forb-dicograph}
\end{figure}

A graph $G$ is a di-cograph if and only if it does not contain one of the
digraphs shown in Fig.\ref{fig:forb-dicograph} as an induced subgraph
\cite{Crespelle:06}. Since each of these graphs contains one of the
forbidden triangles, every Fitch graph is also a di-cograph. On the other
hand, a di-cograph that does not contain $F_1$, $F_5$, or $F_8$ as an
induced subgraph is a Fitch graph. As an immediate consequence of its
characterization in terms of forbidden induced subgraphs, Fitch graphs are
a heritable family, i.e., every induced subgraph of a Fitch graph is again
a Fitch graph. We summarize these observations for later reference as
\begin{lemma} 
  \label{lem:Fitch} 
  The Fitch graphs are a heritable subfamily of the di-cographs.
\end{lemma}
 
A closer inspection shows that four of the eight valid triangles,
$A_1$-$A_4$ can be explained by multiple trees, including one of the
non-binary trees $T_{13}$ to $T_{16}$. In contrast, each of the triangles
$A_5$-$A_8$ with a given labeling of its three leaves is explained by a
unique edge-labeled binary tree, i.e., a specific labeled triple.
\begin{definition} 
  An edge-labeled triple $ab|c$ is \emph{informative} if it explains a
  labeled triangle isomorphic to one of $A_5$, $A_6$, $A_7$ or $A_8$.
\end{definition} 
Thus, if $\X$ contains a triangle of the form $A_5$, $A_6$, $A_7$ or $A_8$
as an induced subgraph, then any tree explaining $\X$ must display the
corresponding informative triple. Any valid relation $\X$ can therefore be
associated with a uniquely defined set $\IT{\X}$ of \emph{informative
  triples} that it displays: $r\in \IT{\X}$ if and only if $r$ is the
unique edge-labeled triple explaining an induced triangle isomorphic to
$A_5$, $A_6$, $A_7$ or $A_8$. For later reference we summarize this fact
as
\begin{lemma}
  If $(T,\lambda)$ explains $\X$, then all triples in $\IT{\X}$ must be
  displayed by $(T,\lambda)$.
  \label{lem:display-inform}
\end{lemma}

\NEW{\section{Least-Resolved Edge-Labeled Phylogenetic Trees}}

In general, there may be more than one rooted (phylogenetic) tree that
explains a given relation $\X$.  In particular, if $\X$ is explained by a
non-binary tree $(T,\lambda)$, then there is always a binary tree
$(T',\lambda')$ that refines $T$ and explains the same relation $\X$ by
setting $\lambda'(e) = \lambda(e)$ for all edges $e$ that are also in $T$
and by choosing the label $\lambda'(e) = 0$ for all edges $e$ that are not
contained in $T$.  \NEW{In this section, we will show that whenever a
  relation $\X$ is explained by an edge-labeled tree $(T,\lambda)$, then
  there exists a unique ``smallest'' tree with this property, which we will
  call the least-resolved tree. These least resolved trees will play a key
  role for obtaining a characterization of Fitch relations in the
  following.}

\begin{definition} \label{def:contract} Let $(T=(V,E),\lambda)$ be an
  edge-labeled phylogenetic tree and let $e=(x,y)\in E$. The phylogenetic
  tree $(T_e, \lambda_e)$, referred to as the (extended) contraction of $e$
  in $(T,\lambda)$, is obtained
  by the following procedure:\\
  First contract the edge $e$ in $T$ and keep the edge-labels of all
  non-contracted edges. If $e$ is an inner edge, the resulting tree is
  again a phylogenetic tree and we are done. The contraction of an outer
  edge $e=(u,v)$, however, leads to (i) the \emph{loss} of a 
  leaf $v$ and (ii) a decrease in the degree of the parental vertex
  $u$. The latter may violate the degree conditions required for a
  phylogenetic tree.  If $u$ is the root of $T$ that has degree $1$ in
  $T_e$, we delete $u$ and its incident edge, and declare the unique
  remaining child of $u$ as the root of $T_e$.  Thus, $T_e$ is obtained by
  an additional simple contraction of the edge
  $(\rho_{T_e},\child(\rho_T))$. Otherwise, if $u$ is an inner vertex that
  has degree $2$ after the contraction of $e$, we apply an additional
  simple contraction of the edge $(u,\child(u))$ \NEW{and set
    $\lambda(\parent(u),u)=1$ if $\lambda(u,\child(u))=1$.}  Equivalently,
  the path from the parent $w$ of $u$ to the unique remaining child $w'$ of
  $u$ is replaced by a single edge $(w,w')$.  This edge is a 1-edge if and
  only if at least one of the edges $wu$ and $uw'$ in the initial tree was
  a 1-edge.
\end{definition}

\begin{definition}
  An edge-labeled phylogenetic tree $(T=(V,E),\lambda)$ is
  \emph{least-resolved (w.r.t.\ $\X_{(T,\lambda)}$)} if none of the
  edge-contracted trees $(T_e,\lambda_e)$, $e\in E$, explains
  $\X_{(T,\lambda)}$.
\end{definition}

It is easy to see that $(T_e, \lambda_e)$ is, by construction, always
obtained by a sequence of simple edge contractions and thus, $(T_e,
\lambda_e)$ is displayed by $(T,\lambda)$.

\begin{lemma}
  \label{lem:contract0}
  Let $(T,\lambda)$ be an edge-labeled phylogenetic  tree.  If $e$ is an
  inner 0-edge in $(T,\lambda)$, then
  $\X_{(T_e,\lambda_e)}=\X_{(T,\lambda)}$. If $e$ is an inner 1-edge, then
  $\X_{(T_e,\lambda_e)}\subseteq \X_{(T,\lambda)}$.
\end{lemma}
\begin{proof}
  The contraction of the inner 0-edge $e=(u,v)$ does not change the number
  of 1-edges along the paths connecting any two leaves. It affects the least
  common ancestor of $x$ and $y$, if $\lca_T(x,y)=u$ or $\lca_T(x,y)=v$. In
  either case, however, the number of 1-edges between the $\lca_T(x,y)$ and
  the leaves $x$ and $y$ remains unchanged. Hence, the relation
  $\X_{(T,\lambda)}$ is not affected by the contraction.

  The contraction of a 1-edge $e$ reduces the number of 1-edges along the path
  between all pairs of leaves whose connecting path in $T$ contain $e$.
  Thus, if $(x,y)\in \X_{(T_e,\lambda_e)}$ then the path connecting $x$ and
  $y$ in $T$ contains also at least one 1-edge, and hence $(x,y)\in
  \X_{(T,\lambda)}$  
\end{proof}

Note that edge contractions therefore always imply
$\X_{(T_e,\lambda_e)}\subseteq \X_{(T,\lambda)}$.
\NEW{  
There may be edges in a tree 
whose  labeling does not affect the relation, i.e., they can
be labeled either 0 or 1.  
The latter observation gives rise to the following definition.

\begin{definition} 
An edge $e$ in a tree $(T,\lambda)$ is  \emph{irrelevant}
if $(T,\lambda')$ with  $\lambda'(e) \neq  \lambda(e)$ and 
 $\lambda'(f) =  \lambda(f)$ for all $f\neq e$ still explains
$\X_{(T,\lambda)}$. 
\end{definition}

Edges that are not 
irrelevant are called \emph{relevant}.
}
As
an example consider the two trees $T_9$ and $T_{12}$ in Figure
\ref{fig:triangles}. Both explain the valid triangle $A_1$. The inner edge
of $T_9$ and $T_{12}$ is a 0-edge and 1-edge, respectively. Thus, this edge
is irrelevant.  The tree $T_{16}$, which is obtained from both $T_9$ and
$T_{12}$ by contracting the irrelevant edge, still explains $A_1$.
\NEW{For later reference, we provide a simple characterization of irrelevant edges.}

\begin{lemma} 
  An edge $e=(u,v)$ is irrelevant in a phylogenetic tree $(T,\lambda)$ if
  and only if $e$ is an inner edge and every path from $v$ to each leaf
  in the subtree rooted at $v$ contains a 1-edge.
\label{lem:relev1}
\end{lemma}
\begin{proof}
  Any inner edge $e$ that satisfies the condition of the lemma is
  irrelevant because every path $u$ to a leaf contains a 1-edge
  irrespective of the label of $(u,v)$.

  Conversely, assume first that $e= (u,v)$ is an outer edge. Hence,
  changing the label of $e$ would immediately change the relation between
  $v$ and any leaf $w$ located in a subtree rooted at a sibling of $v$.
  Since at least one such leaf $w$ exists in a phylogenetic tree, $e$ is
  relevant.  Now suppose that $e= (u,v)$ is an inner edge and that there is
  a leaf $w$ below $v$ such that the path from $v$ to $w$ comprises only
  0-edges. Let $x$ be a leaf such that $\lca(w,x)=u$. Since $T$ is a
  phylogenetic tree, such a leaf always exists. Then $(x,w)\in\X$ if and
  only if $\lambda(e)=1$, i.e., the inner edge $e$ is relevant.   
\end{proof}

A crucial consequence of Lemma~\ref{lem:relev1} is that every outer
edge is relevant. Furthermore, since an irrelevant edge can be relabeled as
a 0-edge without affecting $\X_{(T,\lambda)}$, Lemma \ref{lem:contract0}
implies that irrelevant edges can be contracted without changing
$\X_{(T,\lambda)}$.  These observations naturally pose the question how
edge-labeled trees are structured that cannot be contracted further without
affecting $\X_{(T,\lambda)}$.

\begin{lemma} 
  Let $(T,\lambda)$ be an edge-labeled phylogenetic tree explaining
  $\X$. Then, the tree $(T_e,\lambda_e)$ obtained by contracting the edge
  $e$ explains $\X$ if and only if $e$ is irrelevant or $e$ is an inner
  0-edge.
\label{lem:contract1}
\end{lemma}
\begin{proof}
  The discussion above already shows that irrelevant edges as well as
  0-edges can be contracted without affecting $\X$. We show that
  $\X_{(T_e,\lambda_e)}\neq \X_{(T,\lambda)}$ whenever $e$ is an outer edge
  or a relevant inner 1-edge. First we assume that $e$ is an outer
  edge. Clearly, if $v$ is a leaf, then contracting $e=(u,v)$ would change
  $v$ to an inner vertex in $(T_e,\lambda_e)$. Thus, $L(T)\neq L(T')$ and
  therefore, $(T_e,\lambda_e)$ does not explain $\X$. Now, let $e$ be a
  relevant inner 1-edge. Then, there is a leaf $x$ in the subtree rooted at
  $v$ such that $\Path(v,x)$ consists only of 0-edges (cf.\
  Lemma~\ref{lem:relev1}). Since $(T,\lambda)$ is phylogenetic, there
  exists a leaf $y\in L(T)$ such that $\lca_T(x,y)=u$. Moreover, as
  $\lambda(u,v)=1$, we have $(y,x)\in\X$. Contracting $e$ makes the vertex
  $u^*$, obtained by identifying $u$ and $v$, the least common ancestor of
  $x$ and $y$, i.e., $\lca_{T_e}(x,y)=u^*$. The path from $u^*$ to $x$ now
  contains only 0-edges, i.e., $(y,x)\notin\X_{(T_e,\lambda_e)}$. Thus,
  relevant 1-edges of $(T,\lambda)$ cannot be contracted without affecting
  $\X$.   
\end{proof}

\NEW{The following result shows that relevant edges in a tree $(T,\lambda)$ remain
		relevant in any of its edge-contracted versions  $(T_e,\lambda_e)$, where 
		$e$ is an  inner 0-edge or an irrelevant edge.  }

\begin{lemma}
  \label{lem:conservedrelevance}
  Let $(T,\lambda)$ be an edge-labeled phylogenetic tree explaining $\X$,
  the edge $e$ be an inner 0-edge or an irrelevant 1-edge in $(T,\lambda)$
  and $(T_e,\lambda_e)$ be the tree obtained from $(T,\lambda)$ by
  contracting the edge $e$.  Then, the edge $f\neq e$ is relevant in
  $(T_e,\lambda_e)$ if and only if $f$ is relevant in $(T,\lambda)$.
\end{lemma}
\begin{proof}
  As a consequence of Lemma~\ref{lem:contract1}, $(T_e,\lambda_e)$ still
  explains $\X$.  Lemma~\ref{lem:relev1} implies that the edge $f=(u,v)$ is
  irrelevant in $(T_e,\lambda_e)$ if and only if $f$ is an inner edge
  and all paths from $v$ to leaves below $v$ contain a 1-edge.  If $e$ is
  not located below $f$, then the contraction of $e$ does not affect this
  condition and thus, $f$ is irrelevant in $(T_e,\lambda_e)$ if and only if
  it is irrelevant in $(T,\lambda)$.

  Now suppose $e$ is located below $f$. If $e$ was a 0-edge, the number of
  1-edges along the paths from $v$ to the leaves does not change upon edge
  contraction, and thus $f$ is irrelevant in $(T_e,\lambda_e)$ if and only
  if it is irrelevant in $(T,\lambda)$.  Finally, suppose $e=(u',v')$ was
  an irrelevant 1-edge.  Thus, we can set $\lambda(e)=0$ in $(T,\lambda)$
  without changing the relation $\X$. Now we can repeat the latter
  arguments to conclude that $f$ is irrelevant in $(T_e,\lambda_e)$ if and
  only if it is irrelevant in $(T,\lambda)$.   \end{proof}

\NEW{The following result shows that the order of the contraction of 
		inner 0-edges or irrelevant 1-edges does not affect the resulting relation.}

\begin{lemma}
  \label{lem:ef2}
  Let $(T,\lambda)$ be an edge-labeled phylogenetic tree and let $e$ and
  $f$ be two edges in $E(T)$ such that $(T,\lambda)$, $(T_e,\lambda_e)$ and
  $(T_f,\lambda_f)$ explain the same relation $\X$. Then,
  $((T_e)_f,(\lambda_e)_f)$ obtained from $(T_e,\lambda_e)$ by contracting
  the edge $f$, also explains $\X$.
\end{lemma}
\begin{proof}
  By Lemma \ref{lem:contract1}, an edge can be contracted without affecting
  $\X$ if and only it is an inner 0-edge or an irrelevant 1-edge. The
  labeling of $f$ is not affected by contraction of $e$ and \emph{vice
    versa}. Lemma~\ref{lem:conservedrelevance} furthermore shows that the
  (ir)relevance of an edge $f\ne e$ is conserved by the contraction of
  0-edges and irrelevant 1-edges. Therefore $e$ and $f$ can be contracted
  in arbitrary order and preserve $\X$ in each contraction step.   
\end{proof}

\NEW{We will now apply the results developed so far to least-resolved
  trees. First, we show that the order of edge contractions does not affect
  the resulting least-resolved tree.}

\begin{lemma} 
  Let $(T,\lambda)$ be a least-resolved tree w.r.t.\
  $\X=\X_{(T,\lambda)}$. Then, there is no sequence of edge contractions
  $e_1e_2\dots e_{\ell}$ such that the resulting contracted tree
  $T_{e_1e_2\dots e_{\ell}}$ explains $\X_{(T,\lambda)}$.
  \label{lem:no-cont}
\end{lemma}
\begin{proof}
  Let $(T,\lambda)$ be a least-resolved tree, i.e., none of the
  edge-contracted trees $(T_e,\lambda_e)$, $e\in E$, explains
  $\X_{(T,\lambda)}$.  Lemma \ref{lem:contract0} and \ref{lem:contract1}
  imply that any edge $e\in E$ must be either an outer edge or a relevant
  1-edge.  Clearly, if any edge of the sequence $e_1e_2\dots e_{\ell}$ is
  an outer edge, then the statement is trivially satisfied. 

  Hence, assume that all edges $e_1e_2\dots e_{\ell}$ are inner edges and
  therefore, relevant 1-edges in $(T,\lambda)$.  Lemma~\ref{lem:contract0}
  implies that for $\X$ to change, there must be at least one pair of
  leaves $x,y$ such that $(x,y)\in\X_{(T,\lambda)}$ and
  $(x,y)\notin\X_{(T_e,\lambda_e)}$, i.e., there is no 1-edge along
  $\Path(\lca(x,y),y)$ in $T_e$, and $e$ was the only 1-edge along
  $\Path(\lca(x,y),y)$ in $T$. By Lemma~\ref{lem:contract0},
  $(x,y)\notin\X'$ for the relation explained by any tree that is obtained
  from edge contractions of $(T_e,\lambda_e)$, i.e., there is no sequence
  of edge contractions that leads to a tree $(T',\lambda')$ such that
  $\X_{(T',\lambda')}=\X_{(T,\lambda)}$.   
\end{proof}

\NEW{Next, we summarize some useful properties of least-resolved trees that
  will be used repeatedly in the following sections.}

\begin{lemma}\label{1-edge}
  Let $(T,\lambda)$ be a phylogenetic tree that explains $\X$.  The
  following three conditions are equivalent:
  \begin{enumerate}
  \item \label{it:-1} $(T,\lambda)$ is least-resolved tree w.r.t.\ $\X$.
  \item\label{it:0} Every edge of $(T,\lambda)$ is relevant \NEW{and all
      inner edges are 1-edges.}
  \item\label{it:1}
    (a) Every inner edge  of $(T,\lambda)$ is a 1-edge. \\
    (b) For every inner edge $(u,v)$ there is an outer 0-edge $(v,x)$ in
        $(T,\lambda)$.
  \end{enumerate}
  Moreover, if $(T,\lambda)$ is least-resolved w.r.t.\ $\X$, then
  \begin{enumerate}
    \setcounter{enumi}{3}
  \item\label{it:2} Any inner edge of $(T,\lambda)$ is distinguished by at
    least one informative rooted triple in $\IT{\X}$, and
  \item\label{it:3} For any edge-contracted tree $(T_e, \lambda_e)$ of
    $(T,\lambda)$ there is a triple in $\IT{\X}$ that is not displayed by
    $(T_e, \lambda_e)$, i.e., $(T,\lambda)$ is also least-resolved w.r.t.\
    $\IT{\X}$, and
  \item\label{it:4} The tree $(T(v),\lambda_{|C(v)})$, that is, the subtree
    of $T$ rooted at the vertex $v$ with $\lambda_{|C(v)}(e)=\lambda(e)$
    for any edge $e$ of $T(v)$, is least-resolved w.r.t.\ the subrelation
    $\X_{|C(v)}$ of $\X$ induced by $C(v)$.
  \end{enumerate}
\end{lemma}
\begin{proof}
  The equivalence of Conditions \ref{it:-1} and \ref{it:0} is an immediate
  consequence of Lemma \ref{lem:contract1}.  Moreover, by Lemma
  \ref{lem:contract0}, Condition \ref{it:-1} implies Condition
  \ref{it:1}(a).  To see that also Condition \ref{it:1}(b) is implied given
  Conditions \ref{it:-1} or \ref{it:0}, observe that if $v$ is incident to
  1-edges only, then Lemma~\ref{lem:relev1} implies that $(u,v)$ is
  irrelevant. Thus, $v$ must be incident to at least one 0-edge.  However,
  this 0-edge cannot be an inner edge because inner 0-edges can always be
  contracted due to Lemma \ref{lem:contract0}. Thus, $v$ is incident to an
  outer 0-edge.

  Now assume that Condition \ref{it:1} is satisfied.  First observe that
  none of the outer edges can be contracted without changing $\X$. Let
  $(u,v)$ be an inner 1-edge and $(v,x)$ an outer 0-edge.  Since
  $(T,\lambda)$ is phylogenetic, there is a leaf $y$ for which
  $\lca(x,y)=u$. Thus, $(y,x)\in \X$. However, contraction of the inner
  edge $(u,v)$ would yield $(y,x)\not\in \X$. Thus, none of the inner edges
  can be contracted and therefore, $(T,\lambda)$ is least-resolved w.r.t.\
  $\X$.
   
  \textit{Property \ref{it:2}:} Consider an arbitrary inner edge $e=(u,v)$
  of $T$. Since $(T,\lambda)$ is phylogenetic, there are necessarily leaves
  $x$, $y$, and $z$ such that $\lca(x,y) = v$ and $\lca(x,y,z) = u$. Since
  $(u,v)$ is a 1-edge due to property \ref{it:1}, the tree on $\{x,y,z\}$
  displayed by $T$ must be one of $T_3, T_5, T_7, T_{10}, T_{11}$ or
  $T_{12}$ in Fig.\ \ref{fig:triangles}, where the red inner edge denotes
  the edge $(u,v)$. One easily checks explicitly that neither $T_{11}$ nor
  $T_{12}$ is least-resolved, since contraction of $e=(u,v)$ still yields
  $\X_{(T_e,\lambda_e)} = \X_{(T,\lambda)}$. The remaining trees $T_3$,
  $T_5$, $T_7$, and $T_{10}$, on the other hand, are informative triples
  $xy|z\in \IT{\X}$.  Since $\lca(x,y) = v$ and $\lca(x,y,z) = u$, the edge
  $e$ is by definition distinguished by the triple in $xy|z\in \IT{\X}$.

  \textit{Property \ref{it:3}:} Recall that each inner edge
  $e=(u,v)$ is distinguished by a triple $xy|z\in \IT{\X}$; therefore
  $\lca(x,y) = v$ and $\lca(x,y,z) = u$. However, contraction of $e$ would
  yield $\lca_{T_e}(x,y) = \lca_{T_e}(x,y,z)$, which in turn would imply
  that $xy|z\in \IT{\X}$ is not displayed by $T_e$, a contradiction.
  
  \textit{Property \ref{it:4}:} By construction, no edge $(a,b)$ with
  $v\succeq_T a$ was removed in $T(v)$.  Since
  $\lambda_{|C(v)}(e)=\lambda(e)$ for any edge $e$ of $T(v)$, Property
  \ref{it:1} is trivially fulfilled in $(T(v),\lambda_{|C(v)})$.  Thus,
  $(T(v),\lambda_{|C(v)})$ is least-resolved w.r.t.\ $\X_{|C(v)}$. 
 \end{proof}

As an immediate consequence of Lemma~\ref{lem:ef2}, which implies that all
edge-contractions can be performed independently of each other, we can
observe that for every edge-labeled tree $(T,\lambda)$ there exists a
unique least-resolved tree $(\widehat T,\widehat \lambda)$ that can be
obtained from $(T,\lambda)$ by a sequence of edge-contractions.  Every tree
explaining $\X$ is therefore a refinement of a least-resolved tree that
explains $\X$.  By Lemma \ref{lem:display-inform}, any tree that explains
$\X$ must display the triples in $\IT{\X}$. An even stronger result
holds however:

\begin{lemma}\label{lem:cl}
  If $(T,\lambda)$ be a least-resolved tree w.r.t.\ $\X =
  \X_{(T,\lambda)}$, then $\IT{\X}$ identifies $(T,\lambda)$.
\end{lemma}
\begin{proof}
  If $\IT{\X} = \emptyset$, then, by construction, all induced subgraphs on
  three vertices must be isomorphic to one of the graphs $A_1$, $A_2$,
  $A_3$, or $A_4$ in Fig.\ \ref{fig:triangles}. In this case, $(T,\lambda)$
  is a star-tree, i.e., an edge-labeled tree that consists of outer edges
  only. Otherwise, $(T,\lambda)$ contains inner edges that are, by Lemma
  \ref{1-edge}, distinguished by at least one informative rooted triple in
  $\IT{\X}$, contradicting that $\IT{\X} = \emptyset$.  Hence,
  $r(T)=\emptyset$, and therefore, $r(T) = \cl(\IT{X})$.  Lemma \ref{g1}
  implies that $\IT{\X}$ identifies $(T,\lambda)$.
	    	
  In the case $\IT{\X} \neq \emptyset$, assume for contradiction that $r(T)
  \neq \cl(\IT{\X})$. By Lemma \ref{lem:display-inform} we have $\IT{\X}
  \subseteq r(T)$. Isotony of the closure, Theorem 3.1(3) in
  \cite{Bryant97}, ensures $\cl(\IT{\X}) \subseteq \cl(r(T))=r(T)$. Our
  assumption therefore implies $\cl(\IT{\X}) \subsetneq r(T)$, and thus the
  existence of a triple $ab|c \in r(T)\setminus \cl(\IT{\X})$. In
  particular, therefore, $ab|c \notin\IT{\X}$. Note that neither $ac|b$ nor
  $bc|a$ can be contained in $\IT{\X}$, since $(T,\lambda)$ explains $\X$
  and, by assumption, already displays the triple $ab|c$. Thus, $\IT{\X}$
  contains no triples on $\{a,b,c\}$. 
  
  Lemma \ref{1-edge} implies that there exists a vertex
  $v\in\child(lca(a,b,c))$, with $v\succeq lca(a,b)$, and $(lca(a,b,c),v)$
  is a 1-edge. The subtree $T_{abc}$ of $(T,\lambda)$ with leaves $a,b,c$
  thus corresponds to one of $T_3$, $T_5$, $T_7$, $T_{10}$, $T_{11}$, or
  $T_{12}$ shown in Fig.\ \ref{fig:triangles}. Recall that $T_3$, $T_5$,
  $T_7$, and $T_{10}$ explain the induced subgraphs $A_5$, $A_6$, $A_7$,
  and $A_8$, respectively. If $T_{abc}$ is one of $T_3$, $T_5$, $T_7$, or
  $T_{10}$, then we would have a triple with leaves $a,b,c$ in
  $\IT{\X}$. Since this is not the case by assumption, $T_{abc}$
  must be either $T_{11}$ or $T_{12}$. Thus, the subgraph of $\X$ induced by
  $a,b,c$ is isomorphic to either $A_1$ or $A_4$.

  Moreover, by Lemma \ref{1-edge}, there must be a leaf $d\in\child(v)$
  such that $(v,d)$ is a 0-edge. Hence, the subtrees $T_{acd}$ and
  $T_{bcd}$ with leaves $a,c,d$ and $b,c,d$, respectively, correspond to
  one the trees $T_3$, $T_5$, $T_7$, and $T_{10}$. Thus, the subgraph of
  $\X$ induced by $a,c,d$ or $b,c,d$ must be isomorphic to a valid triangle
  $A_5$, $A_6$, $A_7$ or $A_8$. By construction, $ad|c\in \IT{\X}$ and
  $bd|c\in\IT{\X}$.  Hence, any tree that explains $\X$ must display $ad|c$
  and $bd|c$.  As shown in \cite{Dekker86}, a tree displaying $ad|c$ and
  $bd|c$ also displays $ab|c$.  This implies, however, that $ab|c \in
  \cl(\IT{\X})$, a contradiction to our assumption.

  Therefore, $\cl(R_I)=r(T)$ and we can finally apply Lemma \ref{g1} to
  conclude that $\IT{\X}$ identifies $(T,\lambda)$. 
\end{proof}

We are now in the position to derive the main result of this section. 
\begin{theorem}
  Let $\X \subseteq L\times L$ be a valid relation, $(T,\lambda)$ be a
  phylogenetic tree that explains $\X$ and let $(\widehat T,\widehat
  \lambda)$ be a least-resolved phylogenetic tree w.r.t.\ $\X$.  Then,
  $(T,\lambda)$ displays $(\widehat T,\widehat \lambda)$.  Moreover, the
  tree $(\widehat T,\widehat \lambda)$ has the minimum number of vertices
  among all trees that explain $\X$, and is unique.
  \label{thm:uniqueness}
\end{theorem}
\begin{proof}
  The first statement is an immediate consequence of Lemma~\ref{lem:ef2}.
  Lemma \ref{lem:cl} implies that $\IT{\X}$ identifies $(\widehat
  T,\widehat \lambda)$.  Hence, any tree that displays $\IT{\X}$ is a
  refinement of $(\widehat T,\widehat \lambda)$ and thus, must have more
  vertices.  Lemma \ref{lem:cl} also implies that $(T,\lambda)$ displays
  $(\widehat T,\widehat \lambda)$.  Moreover, Lemma
  \ref{lem:display-inform} ensures that any tree explaining $\X$ 
  display $\IT{\X}$. Combining these two observations, we conclude that
  $\widehat T$ has the minimum number of vertices among all trees that
  explain $\X$.

  By Lemma \ref{1-edge}, all inner and outer edges of $(\widehat T,\widehat
  \lambda)$ are relevant, and thus, their labels cannot be changed without
  changing $\X$. Moreover, Lemma \ref{lem:no-cont} implies that there is no
  further sequence of edge contractions that could be applied to $(\widehat
  T,\widehat \lambda)$ to obtain another tree that explains $\X$.  Hence,
  $(\widehat T,\widehat \lambda)$ is unique.   
\end{proof}

\NEW{\section{Characterization of Valid Xenology Relations}}
 
\NEW{In this section we prove our main result: a binary relation $\X$ is
  explained by a tree if and only if it contains only valid triangles.  The
  key idea of the proof, which proceeds by induction on the number of
  leaves, is to consider the superposition of trees explaining two induced
  subrelations, each of which is obtained by removing a single vertex from
  $\X$. We first establish several technical results for these trees.  To
  this end we introduce some notation that will be used in this section
  only.}
 
\begin{definition}
  \label{def:T-v/contract}
  Let $(T,\lambda)$ be an edge-labeled phylogenetic tree and $e=(u,v)$ be
  an outer-edge of $T$.  We write $(T-v,\lambda_{|L-v})$ for the tree
  obtained from $(T,\lambda)$ by removing the outer edge $e$ and vertex $v$
  from $T$ and keep the edge-labels of all remaining edges.
\end{definition} 
For an outer edge $e=(u,v)$ we therefore have $(T-v,\lambda_{|L-v}) = (T_e,
\lambda_e)$ if and only if either $u=\rho_T$ and $\deg_{T-v}(u)>1$ or
$u\neq\rho_T$ and $\deg_{T-v}(u)>2$.

\begin{definition} 
  \label{def:Xv}
  Let $\X\subset L\times L$ be an irreflexive relation and consider
  $l_1,\dots,l_k\in L$.  The set $\X_{\neg l_1,\dots,l_k}$ denotes the
  subrelation of $\X$ that is induced by $L\setminus \{l_1,\dots,l_k\}$.
\end{definition} 

We emphasize that the results established in the previous sections are in general
not valid for non-phylogenetic trees. Nevertheless, it is useful in the
following to extend some concepts to more general trees. In particular, we
say that an edge-labeled rooted (but possibly non-phylogenetic) tree
$(T,\lambda)$ \emph{explains} a given irreflexive relation $\X$ if for any
pair $(x,y)\in \X$ there is a 1-edge on the path from $\lca(x,y)$ to $y$.

Using the same arguments as in the proof of Lemma~\ref{lem:induced} we
observe that $(T-v,\lambda_{|L-v})$ explains $\X_{\neg v}$.

\begin{lemma}\label{del}
  Let $(T,\lambda)$ be a least-resolved phylogenetic tree on $L$ w.r.t.\
  $\X = \X_{(T,\lambda)}$, and $v\in L$.  Let $(T',\lambda')$ be a
  least-resolved phylogenetic tree w.r.t.\ $\X_{\neg v}$.  Then,
  $(T',\lambda')$ is displayed by $(T-v,\lambda_{|L-v})$.  In particular,
  $(T',\lambda')=(T-v,\lambda_{|L-v})$ if and only if (i)
  $\parent(v) = \rho_T$ and $\deg_T(\rho_T)>2$ or (ii)
  $\deg_T(\parent(v))>3$ and $\lambda_{|L-v}(\parent(v),u)=0$ for some
  child $u\in\child(\parent(v))$, $u\neq v$.
\end{lemma}
\begin{proof}	
  Let $(T',\lambda')$ be least-resolved w.r.t.\ $\X_{\neg v}$. If
  $(T-v,\lambda_{|L-v})$ is phylogenetic, then we may apply Thm.\
  \ref{thm:uniqueness} to verify that $(T',\lambda')$ is indeed displayed
  by $(T-v,\lambda_{|L-v})$. Now assume that $(T-v,\lambda_{|L-v})$ is not
  phylogenetic.  In this case, either (a) $\parent(v)\ne\rho_T$ is an
  inner vertex of degree $2$, or (b) the root $\rho$ of $T-v$ has degree
  $1$, and hence $\rho_T=\parent(v)$.
	
  \emph{Case (a):} If $x = \parent(v)\ne\rho_T$ is an inner vertex of
  degree $2$, let $T^*$ be the tree obtained by a simple contraction of the
  edge $(\parent(x), x)$ and setting $\lambda_{|L-v}(x,\child(x))=1$. The
  labels of all other edges are kept. By construction, we obtain a
  phylogenetic tree $(T^*,\lambda^*)$ that still explains $\X_{\neg v}$ and
  satisfies $(T',\lambda') \le
    (T^*,\lambda^*)\le(T-v,\lambda_{|L-v})$.  Therefore, $(T',\lambda')$
  is displayed by $(T-v,\lambda_{|L-v})$.

  \emph{Case (b):} If the root $\rho$ of $T-v$ has degree $1$, let $T^*$ be
  the tree obtained by deleting $\rho$ and the edge $(\rho,w)$, where $w$
  denotes the unique child of $\rho$ in $T-v$, and declaring $\child(\rho)$
  as the root of $T'$. For all other edges set
  $\lambda^*(e)=\lambda_{|L-v}(e)$.  Again, we obtain a phylogenetic tree
  $(T^*,\lambda^*)$ that still explains $\X_{\neg v}$.  Repeating the
  arguments of \emph{Case (a)}, we can conclude that $(T',\lambda')$ is
  displayed by $(T-v,\lambda_{|L-v})$.
			
  \smallskip 

  Now assume that $(T',\lambda')=(T-v,\lambda_{|L-v})$.  There are two
  cases: either $\parent(v)$ is the root $\rho_T$ or not. If $\parent(v) =
  \rho_T$, then $\deg_T(\rho_T) \leq 2$ would imply that $\deg_{T'}(\rho_T)
  \leq 1$, in which case $(T',\lambda')$ would not be a phylogenetic tree;
  a contradiction, since $(T',\lambda')$ is phylogenetic.  Hence, if
  $\parent(v) = \rho_T$, then $\deg_T(\rho_T) > 2$.  Now assume that
  $\parent(v) \neq \rho_T$. Thus, there is an inner edge $(x,\parent(v))$
  where $x=\parent(\parent(v))$.  Lemma \ref{1-edge}(\ref{it:1}) implies
  that this edge $(x,\parent(v))$ must be incident to an outer 0-edge in
  $(T',\lambda')$ and hence, $\lambda_{|L-v}(\parent(v),u)=0$ for some leaf
  $u\in L\setminus\{v\}$. Moreover, as $(T',\lambda')$ is phylogenetic,
  $\deg_{T-v}(\parent(v))>2$ and hence, $\deg_T(\parent(v))>3$.

  Conversely, assume first that $\parent(v) = \rho_T$ and
  $\deg_T(\rho_T)>2$.  In this case, $(T-v,\lambda_{|L-v})$ is still a
  phylogenetic tree.  By construction, $E^0(T-v)=E^0(T)$ and
  $\lambda_{|L-v}(e)=\lambda(e)$ for all $e\in E^0(T-v)$ Thus, any inner
  edge of $T-v$ is a 1-edge.  Lemma \ref{1-edge}(\ref{it:1}) implies that
  for each inner edge $e=(x,y)$ in $T$ there is an outer 0-edge $(y,z)$ in
  $(T,\lambda)$. This property still holds in $(T-v,\lambda_{|L-v})$
  because the deleted edge $(\parent(v),v)$ is incident to the root of
  $(T,\lambda)$. Thus all edges of $(T-v,\lambda_{|L-v})$ are relevant.
  Lemma \ref{1-edge} implies that $(T-v,\lambda_{|L-v})$ is least-resolved.

  Now assume that $\parent(v) \neq \rho_T$ and $\deg_T(\parent(v))>3$.
  Thus, $(T-v,\lambda_{|L-v})$ is still a phylogenetic tree.  Let
  $\lambda_{|L-v}(\parent(v),u)=0$ for some child $u\in\child(\parent(v))$,
  $u\neq v$. Now, we can apply similar arguments as above to conclude that
  all edges in $(T-v,\lambda_{|L-v})$ are relevant, and thus,
  $(T-v,\lambda_{|L-v})$ is least-resolved.

  In summary, if $\parent(v) = \rho_T$ and $\deg_T(\rho_T)>2$ or
  $\lambda_{|L-v}(\parent(v),u)=0$ for some child $u\in\child(\parent(v))$,
  $u\neq v$, and $\deg_T(\parent(v))>3$, then $(T-v,\lambda_{|L-v})$
  is least-resolved w.r.t.\ $\X_{\neg v}$. By Thm.~\ref{thm:uniqueness},
  $(T-v,\lambda_{|L-v}) = (T',\lambda')$.   
\end{proof}

An immediate consequence of Lemma \ref{del} is the following result
\NEW{that is crucial for proving the main result.}

\begin{lemma}\label{del2}
  Let $(T,\lambda)$ and $(T-v,\lambda_{|L-v})$ be defined as in Lemma
  \ref{del}, and $(T',\lambda')$ be the least-resolved phylogenetic tree
  that explains $\Xv$.  Then, either
  \begin{enumerate}
  \item $(T-v,\lambda_{|L-v}) = (T',\lambda')$, or
  \item $(T',\lambda')$ is obtained from $(T-v,\lambda_{|L-v})$ by a simple
    contraction of either
    \begin{description}
    \item[(i)] the inner edge $(\rho_T,u)\in E(T-v)$, in case that
      $\parent(v) = \rho_T$ and $\deg_T(\rho_T) = 2$, or
    \item[(ii)] the inner edge $(\parent(x), x)\in E(T-v)$, where
      $x=\parent(v)\neq \rho_T$, and setting $\lambda'(x,\child(x))=1$,
      otherwise.
    \end{description}
    In either case $\lambda'(e)=\lambda_{|L-v}(e)$ for all non-contracted
    edges $e$.
  \end{enumerate}
  In particular, $(T-v,\lambda_{|L-v})$ displays the least-resolved
  phylogenetic tree $(T',\lambda')$ that explains $\Xv$ and
  therefore, $r(T')\subseteq r(T-v)$.
\end{lemma}
\begin{proof}
  By Lemma \ref{del}, $(T-v,\lambda_{|L-v})$ is least-resolved if and only
  if $\parent(v) = \rho_T$ and $\deg_T(\rho_T)>2$ or there exists a leaf
  $u\in \parent(v)$, $u\neq v$, such that $\lambda_{|L-v}(\parent(v),u)=0$
  and $\deg_T(\parent(v))>3$.  If $(T-v,\lambda_{|L-v})$ is not
  least-resolved and $\parent(v) = \rho_T$, we have $\deg_{T-v}(\rho_T)=1$.
  Due to Lemma \ref{1-edge}(\ref{it:4}), the tree $(T',\lambda')$ obtained
  by a simple contraction of the single edge $(\rho_T,u)$ and adopting $u$
  as the new root is least-resolved w.r.t.\ $\Xv$.
 
  If $(T-v,\lambda_{|L-v})$ is not least-resolved and $\parent(v) \neq
  \rho_T$, then either (a) there is no leaf $u\in\child(\parent(v))$,
  $u\neq v$, with $\lambda_{|L-v}(\parent(v),u)=0$ or (b)
  $\deg_{T-v}(\parent(v))=2$.  Indeed, $\deg_{T-v}(\parent(v))>2$ and
  $u\in\child(\parent(v))$ with $\lambda_{|L-v}(\parent(v),u)=0$ implies
  that $(T-v,\lambda_{|L-v})$ is least-resolved. On the other hand,
  $\deg_{T-v}(\parent(v))\ge2$ because $T$ is phylogenetic.
  
  Case (a).  Assume that $(\parent(v),u)$ is a 1-edge for all children
  $u\ne v$ of $\parent(v)$.  Then, the inner edge $(\parent(x),x)\in
  E(T-v)$ is irrelevant in $(T-v,\lambda_{|L-v})$; thus it can be
  contracted. Since $(T, \lambda)$ is least-resolved, Lemma
  \ref{1-edge}(\ref{it:1}) ensures that every inner vertex in $(T-v,
  \lambda_{|L-v})$ other than $\parent(v)$ is adjacent to an outer
  0-edge. Hence, contraction of $(x,\parent(v))$ in $(T-v,\lambda-v)$
  yields the least-resolved tree w.r.t.\ $\Xv$.
  
  Case (b). If $\deg_{T-v}(\parent(v))=2$ and
  $\lambda_{|L-v}(\parent(v),u)=1$, the edge $(\parent(x),x)$ can be
  contracted without changing the relation and similar arguments as in case
  (a) show that $(T-v,\lambda_{|L-v})$ is least-resolved w.r.t.\ $\Xv$. If
  $\lambda_{|L-v}(\parent(v),u)=0$, then the construction as in 2.(ii) does
  not change $\X$ since $\lambda(\parent(x),x)=1.$ Again, similar arguments
  as in case (a) ensure that $(T-v,\lambda_{|L-v})$ is least-resolved
  w.r.t.\ $\Xv$.

  Obviously, either $(T',\lambda')=(T-v,\lambda_{|L-v})$ or $(T',\lambda')$
  can be obtained from $(T-v,\lambda_{|L-v})$ by a single simple
  edge-contraction. Thus $(T',\lambda')$ is displayed by
  $(T-v,\lambda_{|L-v})$ and $r(T')\subseteq r(T-v)$.   
\end{proof}

Let $(T=(V,E),\lambda)$ be an edge-labeled phylogenetic tree. Moreover, let
$(x,y)\in E$ and let $(T_e,\lambda_e)$ be the phylogenetic tree obtained
from $(T,\lambda)$ by extended contraction of $e$ in $(T,\lambda)$. Given
$(T_e,\lambda_e)$ it is possible to recover the tree $(T,\lambda)$
reverting the extended contraction of $e$. If $e$ was an internal edge,
this amounts to subdividing a vertex $z$, yielding $e=(u,v)$, and a
bi-partitioning of the set of children of $z$ into the children of $u$ and
$v$. If $e$ was an external edge incident to a degree 2 node, an edge $f$
in $(T_e,\lambda_e)$ is subdivided and $e$ is attached to the new inner
vertex. In addition, the labeling is adjusted.  We refer to these
constructions as \emph{reinsertion of $e$ into $(T_e,\lambda_e)$}.

\begin{lemma}
  Given a Fitch relation $\X$ such that $\X_{\neg u}$, $\X_{\neg v}$, and
  $\X_{\neg uv}$ are valid for some $u,v\in V(\X)$.  Let $(\Tu,\lu)$,
  $(\Tv,\lv)$ and $(\Tuv,\luv)$ be the least-resolved trees that explain
  $\X_{\neg u}$, $\X_{\neg v}$, and $\X_{\neg uv}$, respectively.

  Then, there is a tree $(T,\lambda)$ that correctly explains all members
  in $\X\setminus \X[u,v]$, i.e., $\X_{(T,\lambda)}[x,y] = \X[x,y]$ for all
  $x,y$ with $\{x,y\}\neq \{u,v\}$. Moreover $(T,\lambda)$ displays
  $(\Tu,\lu)$, $(\Tv,\lv)$ and $(\Tuv,\luv)$.
\label{lem:Xuv} 
\end{lemma}
\begin{proof}
  Consider the least-resolved tree $(\Tuv,\luv)$ that correctly explains
  $\X_{\neg uv}$. By Lemma \ref{del2}, this tree can be obtained from the
  least-resolved trees $(\Tu,\lu)$ and $(\Tv,\lv)$ by removing the vertices
  $v$ and $u$, respectively, and possibly contraction of edges.  More
  precisely, $(\Tuv,\luv) = (\Tu-v,\lambda_{\neg u|L'})$ or $(\Tuv,\luv)$
  is obtained from $(\Tu-v,\lambda_{\neg u|L'})$ by contracting
  \emph{exactly} the edge $(x,y)$ where $y=\parent(v)$ and a possible
  relabeling of the children of $y$. In what follows, we denote by $xy$
  the vertex in $T_{uv}$ that is obtained by contraction of this edge
  $(x,y)$.  In the same way, $(\Tuv,\luv)$ is obtained from
  $(\Tv-u,\lambda_{\neg v|L'})$ and if the edge $(x',y')$ was contracted,
  then $x'y'$ denotes the resulting vertex in $\Tuv$.

  Therefore, the following cases must be considered:
  \begin{enumerate}
  \item
    $(\Tuv,\luv) = (\Tu-v,\lambda_{\neg u|L'}) =(\Tv-u,\lambda_{\neg v|L'})$.
  \item Either 
    \begin{itemize}
    \item[(a)] $(\Tuv,\luv) = (\Tu-v,\lambda_{\neg u|L'}) 
      \lneq (\Tv-u,\lambda_{\neg v|L'})$, or 
    \item[(b)] $(\Tuv,\luv) = (\Tv-u,\lambda_{\neg v|L'}) 
      \lneq (\Tu-v,\lambda_{\neg u|L'})$.
    \end{itemize}
  \item  $(\Tuv,\luv) \lneq (\Tu-v,\lambda_{\neg u|L'})$ 
    and $(\Tuv,\luv) \lneq (\Tv-u,\lambda_{\neg v|L'})$
    and either\\ 
    (a) $xy\neq x'y'$ or (b) $xy =x'y'$.
  \end{enumerate}
	
  In \textbf{Case 1}, one can simply add the edge $(\parent(v),v)$ and
  $(\parent(u),u)$ together with the original edge labels
  $\lu(\parent(v),v)$ and $\lv(\parent(u),u)$ to obtain a tree
  $(T,\lambda)$ that contains both $(\Tu,\lu)$ and $(\Tv,\lv)$ as subtrees
  and thus, $\X_{(T,\lambda)}[x,y] = \X[x,y]$ for all $x,y$ with
  $\{x,y\}\neq \{u,v\}$.

  In \textbf{Case 2(a)}, one can simply add the edge $(\parent(v),v)$
  together with the original edge label $\lu(\parent(v),v)$ to obtain
  $(\Tu,\lu)$. Since $x'y'$ denotes the vertex that results from
  contracting the edge $(x',y')$ in $(\Tv-u,\lambda_{\neg v|L'})$, this
  vertex is also contained in $(\Tu,\lu)$.  Now, we reinsert $x'y'$ such
  that we obtain a tree $(T,\lambda)$ that contains $(\Tv,\lv)$ as a
  subtree. Hence, $\Xuv$ and $\X_{\neg v}$ are correctly explained by
  $(T,\lambda)$. It remains to show that also all $\X[v,z]$ and $\X[z,v]$
  with $z\neq u$ are still correctly explained. Assume for contradiction
  that this is not the case and that $\X[v,z] \neq \X_{(T,\lambda)}[v,z]$
  for some $z\neq u$.  This is only possible if in the tree $(T,\lambda)$
  there is this 1-edge $(x',y')$ contained in the path from $\lca_T(v,z)$
  to $z$. Hence, $\X_{(T,\lambda)}[v,z] = (v,z)$, which implies that the
  path from $\lca_{\Tu}(v,z)$ to $z$ contains only 0-edges. Moreover, $\Tu$
  is least-resolved w.r.t.\ $\Xu$. Hence, all inner edges are
  1-edges. Therefore, $\lca_{\Tu}(v,z) = x'y'$ and $(x'y', z)\in E(\Tu)$
  must be an outer 0-edge.  Note that this implies that $z$ is a child of
  $y'$ in $\Tv$. By construction according to Lemma \ref{del2}(ii), we have
  contracted the edge $(x',y')$ in $(\Tv-u,\lambda_{|L-u})$ and relabeled
  all outer edges in $\Tv$ incident to $y'$ as 1-edges. But this implies
  that $(x'y',z)$ is a 1-edge in $\Tu$; a contradiction.  The assumption
  $\X[z,v] \neq \X_{(T,\lambda)}[z,v]$ for some $z\neq u$ yields a
  contradiction using analogous arguments.
	
  \textbf{Case 2(b)} is settled by interchanging the roles of $u$ and $v$
  in Case 2(a).

\begin{figure}[htbp]
  \begin{center}
    \includegraphics[width=0.9\textwidth]{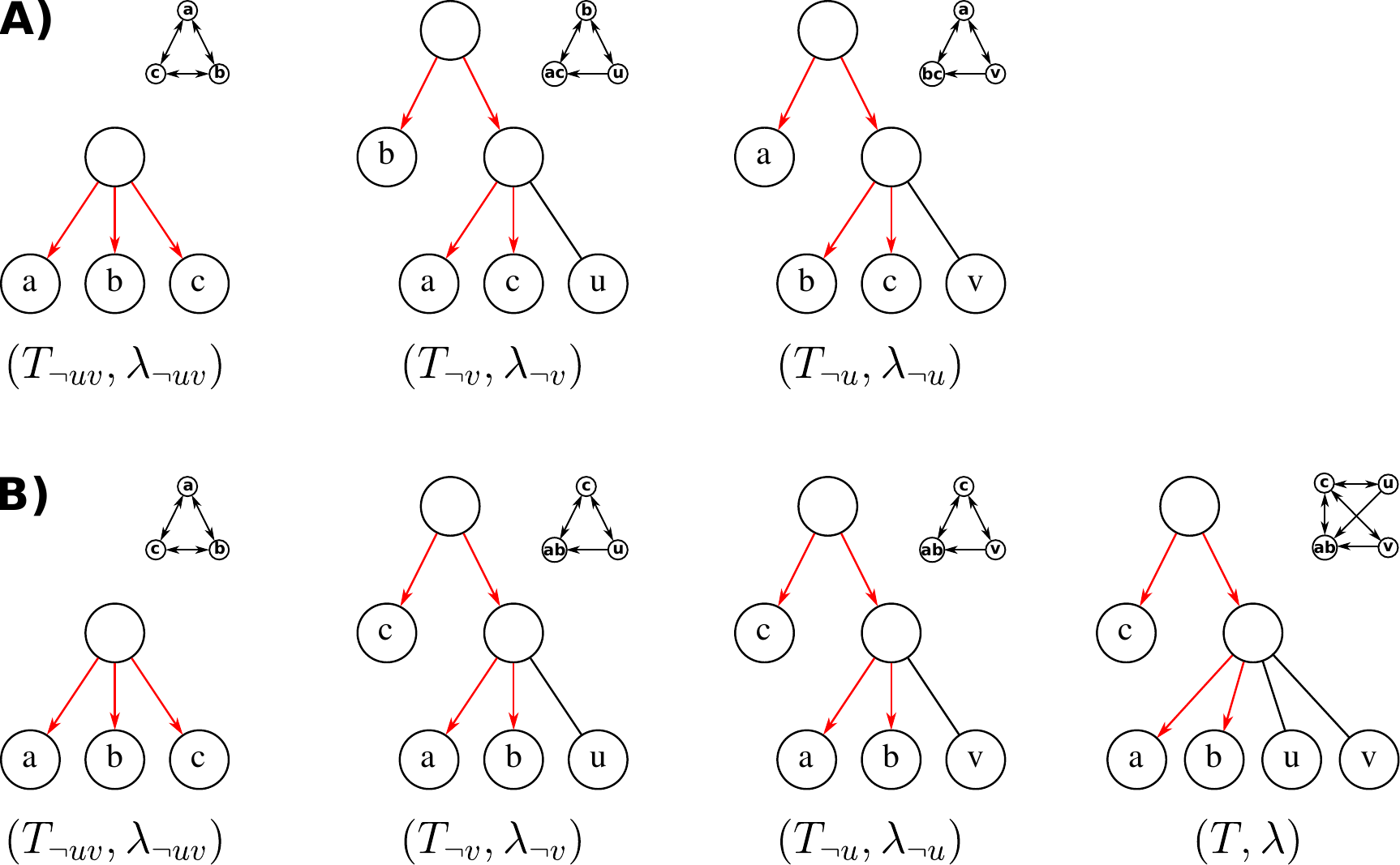}
  \end{center}
  \caption[]{(A) The two least-resolved trees $(\Tu,\lu)$ and $(\Tv,\lv)$
    that explain $\Xu$ and $\Xv$ respectively both explain the
    least-resolved tree $(\Tuv,\luv)$ that explains $\Xuv$. However, there
    exists no tree $(T,\lambda)$ that explains $\X$, thus there is no valid
    Fitch relation $\X$ that contains both $\Xu$ and $\Xv$. This is due to
    the fact that the triples $ac|b$ and $bc|a$ in $(\Tu,\lu)$ and
    $(\Tv,\lv)$ contradict each other.  (B) We have
    $\C_x\NEW{=\{c\}}=\C_{x'}$ and $\C_y\NEW{=\{a,b\}}=\C_{y'}$ in the
    least-resolved trees $(\Tu,\lu)$ and $(\Tv,\lv)$. In this case, there
    exists a tree $(T,\lambda)$ that displays $(\Tuv,\luv)$, $(\Tu,\lu)$
    and $(\Tv,\lv)$, and explains $\X$.  \NEW{The Fitch relation
      corresponding to each tree is shown in the upper right corner. Two
      nodes $x$ and $y$ are represented as one node $xy$ if they have the
      same relationship with every other node.}  }
  \label{fig:Xuv-counterex}
\end{figure}

  \textbf{Case 3.} In order to obtain $(T,\lambda)$ from $(\Tuv,\luv)$, we
  need to undo the contractions that lead to $xy$ and $x'y'$ and in
  addition, reinsert the edges $(y',u)$ and $(y,v)$ with original
  edge-labeling such that $(T,\lambda)$ contains both $(\Tu,\lu)$ and
  $(\Tv,\lv)$ as subtrees and thus, $\X_{(T,\lambda)}[x,y] = \X[x,y]$ for
  all $x,y$ with $\{x,y\}\neq \{u,v\}$. The subdivision of $xy$ partitions
  the set of children $\child(xy)$ of the vertex $xy$ into two disjoint
  sets $\C_x$ and $\C_y$ in such a way that $\C_x$ contains all children of
  $x$ that are distinct from $y$ and $\C_y$ contains all children of $y$ in
  $(\Tu-v,\lambda_{\neg u|L'})$. Analogously, the sets $\C_{x'}$ and
  $\C_{y'}$ are obtained by partitioning $\child(x'y')$ in
  $(\Tv-u,\lambda_{\neg v|L'})$. The sets $\C_{x}$, $\C_{x'}$, $\C_{y}$,
  and $\C_{y'}$ are all non-empty because $(\Tu,\lu)$ and $(\Tv,\lv)$ are
  phylogenetic.
	
  \textbf{Case 3(a).} $xy\neq x'y'$. By definition of $(\Tuv,\luv)$, it is
  possible to subdivide $xy$ and add $(\parent(v),v)$ with the
  edge-labeling $\lu(\parent(v),v)$ such that we obtain $(\Tu,\lu)$.
  Subdivision of $x'y'$ in $(\Tu,\lu)$ results in a tree $(T,\lambda)$ that
  contains $(\Tv,\lv)$ as a subtree. Hence, $(T,\lambda)$ correctly
  explains $\Xuv$ and $\Xv$.  Arguments analogous to Case 2 now show that
  $\X[z,v]$ and $\X[v,z]$ are correctly explained for any $z\neq u$, thus
  $(T,\lambda)$ correctly explains $\Xu$.
  
  \textbf{Case 3(b).} $xy=x'y'$.  Since $xy=x'y'$, $(T,\lambda)$ is
  obtained from $(\Tuv,\luv)$ by reinsertion of a single edge. To ensure
  that $(T,\lambda)$ displays both $(\Tu,\lu)$ and $(\Tv,\lv)$, we need to
  show that $\C_x=\C_{x'}$ and $\C_y= \C_{y'}$.
	
  First, we show that all 0-edges incident to $xy$ in $(\Tuv,\luv)$ are
  incident to $x$ and $x'$ in $(\Tu,\lu)$ and $(\Tv,\lv)$, respectively.
  Let $M$ denote the set of all leaves $z\in\child(xy)$ for which
  $\lambda'(xy,z)=0$ in $\Tuv$. Since $(\Tuv,\luv)$ is least-resolved,
  $M\ne\emptyset$. For any $w\in\child(xy)$, and $z\in M$ there is no
  1-edge on the path from $\lca(w,z)$ to $z$ in $(\Tuv,\luv)$. We proceed
  by showing that $M \subseteq \C_x\cap\C_{x'}$.  Assume for contradiction
  that $z\in\C_x$ but $z\not\in\C_{x'}$. Thus $z\in\C_{y'}$. Furthermore,
  for any $w'\in\C_{x'}$, the 1-edge $e'=(x',y')$ is contained in the path
  from $\lca(w',z)$ to $z$ in the tree $(\Tv-u,\lambda_{v|L'})$. Since
  $(\Tv-u,\lambda_{v|L'})$ is phylogenetic, $\C_{x'}$ is non-empty, i.e.,
  such a $w'$ exists. In contrast, for any $w\in \C_x\cup\C_y$, $w\ne z$,
  there is no 1-edge on the path from $\lca(w,z)$ to $z$ in
  $(\Tu-v,\lambda_{u|L'})$.  Since $\C_{x'} \subseteq \C_x\cup\C_y$, the
  two trees $(\Tu-v,\lambda_{u|L'})$ and $(\Tv-u,\lambda_{v|L'})$ cannot
  explain the same relation $\Xuv$; this is the desired contradiction.
	
  Hence, it remains to show that for every 1-edge $(xy,a)$ in $(\Tuv,\luv)$
  either $a\in \C_x\cap\C_{x'}$ or $a\in\C_y\cap\C_{y'}$ is true. Assume
  for contradiction that $a\in \C_x$ but $a\notin\C_{x'}$, i.e., $a\notin
  \C_y$ and $a\in\C_{y'}$. This implies $[a,v]\in\Xu$ and
  $(u,a)\in\Xv$. Since $\{u,v,a\}$ must form a valid triangle, either
  $(u,v)\in\X$ or $[u,v]\in\X$ must be true. On the other hand, since $M
  \subseteq \C_x\cap\C_{x'}$ and the trees $(\Tu,\lu)$ and $(\Tv,\lv)$ are
  least-resolved, both $(y,v)$ and $(y',u)$ must be 0-edges. By
  construction, $(T,\lambda)$ is obtained by reinserting a single edge in
  $(\Tuv,\luv)$ in such a way that $\parent(u)=\parent(v)$. Thus we must
  have $u|v$; a contradiction, and we can conclude $\C_y=\C_{y'}$ and
  $\C_x=\C_{x'}$.  
\end{proof}

We remark that the existence of the tree $(T,\lambda)$ asserted in
Lemma~\ref{lem:Xuv} does not follow from the fact that both $(\Tu,\lu)$ and
$(\Tv,\lv)$ explain $(\Tuv,\luv)$. A counter-example is given in
Fig.~\ref{fig:Xuv-counterex}. The condition that the trees together explain
a Fitch relation cannot be relaxed in the proof.

\begin{theorem}\label{existence}
  An irreflexive relation $\X$ on $L$ is valid if and only if it is a Fitch
  relation.
\label{thm:main}
\end{theorem}
\begin{proof}
  Assume that $\X$ is valid.  Hence, there is a tree $(T, \lambda)$ that
  explains $\X$. Let $x, y, z \in L$ be distinct vertices. Clearly, any
  subtree $T' \subseteq T$ with leaf set $\{x, y, z\}$ must correspond to
  one of the trees $T_1,\dots,T_{16}$ in Fig.\ \ref{fig:triangles}.  Since
  these subtrees can only encode the valid triangles $A_1,\dots A_8$, the
  subgraph induced by $x, y, z$ in $\X$ must be isomorphic to one of
  $A_1,\dots A_8$.  Since this statement is true for any three distinct
  vertices in $\X$, all triangles in $\X$ are valid. Hence, $\X$ is a Fitch
  relation.
	
  Now assume that $\X$ is a Fitch relation.  The trivial relation on $L$,
  corresponding to the empty graph, is explained by any tree with leaf set
  $L$ that has only 0-edges. For the non-trivial case we proceed by
  induction w.r.t.\ the number of vertices $|L|$. The base case consists of
  the valid triangles, for which the statement is trivially true. Assume
  now that all Fitch relations with $|L|\leq n$ are valid.

  Let $\X$ be a Fitch relation on $|L|=n+1$ vertices and let $u,v\in L$ be
  two distinct, arbitrarily chosen vertices.  Clearly, $\Xu$, $\Xv$, and
  $\Xuv$ are Fitch relations and, by assumption, also valid. In particular,
  there are unique least-resolved trees $(\Tu,\lu)$, $(\Tv,\lv)$ and
  $(\Tuv,\luv)$ that explain $\Xu$, $\Xv$ and $\Xuv$, respectively. With
  the exception of the relation between $u$ and $v$, $\X$ is therefore
  determined by $(\Tu,\lu)$ and $(\Tv,\lv)$, i.e., any pair $(x,y)\in \X$
  for which $\{x,y\}\ne\{u,v\}$ is explained by $(\Tu,\lu)$ or
  $(\Tv,\lv)$. In particular all pairs $(x,u)$ or $(u,x)$ in $\X\setminus
  \X[u,v]$ are explained by $(\Tv,\lv)$ and all pairs $(x,v)$ or $(v,x)$ in
  $\X\setminus \X[u,v]$ are explained by $(\Tu,\lu)$.

  Lemma \ref{lem:Xuv} implies that there is a tree that correctly explains
  all pairs in $\X\setminus \X[u,v]$ and displays $(\Tu,\lu)$, $(\Tv,\lv)$
  and $(\Tuv,\luv)$. Thus, there is in particular a least-resolved tree
  $(T,\lambda)$ that fulfills these requirements.
	
  $\X[u,v]$ is in some cases uniquely determined by $\X\setminus \X[u,v]$
  and the requirement that $\{u,v,x\}$ forms a valid triangle.  The
  existence of $(T,\lambda)$ then implies immediately that $\X[u,v]$, and
  hence $\X$, is explained by $(T,\lambda)$.

  This is not always the case, however. If more than one choice of
  $\X[u,v]$ completes $\X\setminus \X[u,v]$, we need to show that a
  $(T,\lambda)$ exists for each of the possible choices.  Denote by
  $\Delta_{uv}$ the set of triangles in $\X$ that contain $u$ and
  $v$. Full enumeration (which we leave to the reader) shows that $\X[u,v]$
  is not uniquely determined if and only if all triangles in $\Delta_{uv}$
  are of the form $A$, $B$, $C$ or $D$ listed in Fig.\ \ref{uRv}. Only certain
  combinations of these triangle types can occur: The co-occurence of $A$
  and $B$ implies $(u,v)\in\X$, hence $\X[u,v]$ is uniquely determined, and
  hence $(T,\lambda)$ is also unique. The remaining cases can be classified
  as follows:
  \begin{enumerate}  
  \item $\Delta_{uv}$ contains at least one triangle either of types $A$,
    $C$ and $D$ but not $B$, or of types $B$, $C$ and $D$ but not $A$.
  \item $\Delta_{uv}$ consists of triangles of exactly one of the types
    $A$, $B$ and $C$, $D$, respectively, and for each type there is a
    triangle.
  \item $\Delta_{uv}$ consists exclusively of triangles of the types $C$
    and $D$ and for each type $C$, $D$ there is a triangle.
  \item All triangles in $\Delta_{uv}$ are of the same type. 
  \end{enumerate}

  \begin{figure}
    \begin{center}
      \includegraphics[width=\textwidth]{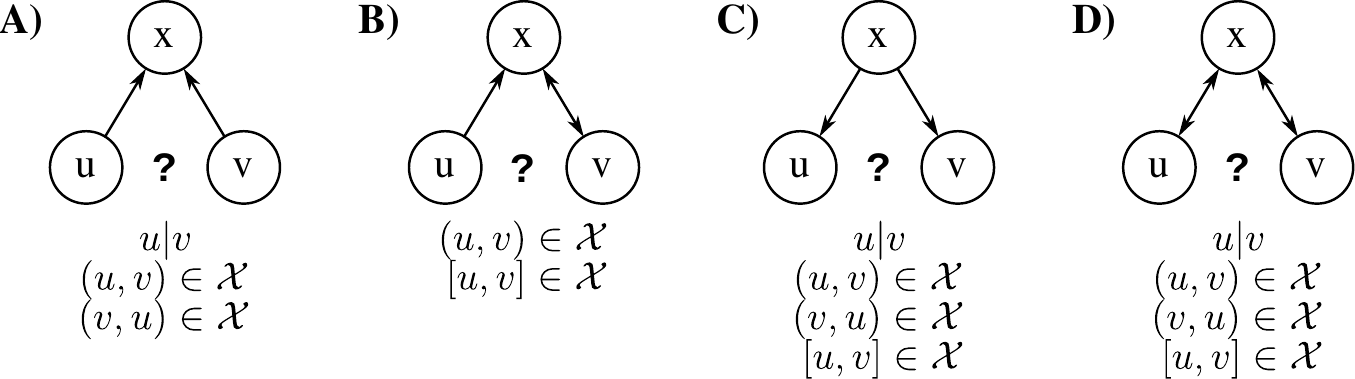}
    \end{center}
    \caption[]{All cases where the relationship $\X[u,v]$ cannot be
      uniquely inferred from $\X\setminus \X[u,v]$. }
    \label{uRv}
  \end{figure}

  In each of these cases, there is more than one possible choice for
  $\X[u,v]$. Lemma~\ref{lem:Xuv} ensures that there is a least-resolved
  tree $(T,\lambda)$ that explains at least one of these choices. Given
  $(T,\lambda)$ for one particular choice, we show below that it is always
  possible to transform $(T,\lambda)$ into another least-resolved tree that
  explains $\X$ with a different choice of $\X[u,v]$. The resulting tree
  $(T',\lambda')$ is unique by Thm.\ \ref{thm:uniqueness}, and thus the
  transformation can be inverted in a uniquely defined manner.

  In what follows, we call $v\in L(T)$ a \emph{sink} (resp.\ \emph{source})
  if for all $x\in L(T)$ we have $(x,v)\in \X_{(T,\lambda)}$ (resp.\
  $(v,x)\in \X_{(T,\lambda)}$).  Moreover, in order to exclude the trivial case
  $\Delta_{uv} = \emptyset$, we assume that $|L(T)|\geq 3$.

  \textbf{Case 1a.} 
  Suppose that $\Delta_{uv}$ contains at least one triangle each of types
  $A$, $C$ and $D$ but not of type $B$.  Hence, $\X[u,v] \in
  \{(u,v),(v,u), u|v\}$.  Thus, we show that for each of these choices of
  $\X[u,v]$ there is a least-resolved tree that explains $\X$.
	
  Suppose that $(T,\lambda)$ explains $u|v$. Since all inner edges of
  $(T,\lambda)$ are 1-edges, $u$ and $v$ must be siblings and in
  particular, the edges $(\lca(u,v),v)$ and $(\lca(u,v),u)$ are
  0-edges. Since there is a triangle of each of the types $A$, $C$ and $D$,
  there is no leaf $x\in L(T)\setminus \{u,v\}$ with $x|u$ or $x|v$.
  Moreover, if there would be another vertex $x\in V(T)\setminus \{u,v\}$
  that is adjacent to $\lca_T(u,v)$, then $(\lca(u,v),x)$ must be a 1-edge.
  Let $T^*$ denote the subtree of $T$ with root $\lca(u,v)$ without the
  leaves $u$ and $v$. Therefore, $(T,\lambda)$ locally looks like the tree
  shown in the first panel in Fig.\ \ref{cases}(1a).  In order to obtain a
  tree that explains $(u,v)$, we can modify $(T,\lambda)$ locally to obtain
  a tree $(T',\lambda')$ by inserting a single inner 1-edge $(a,b)$ in such
  a way that $b$ becomes the new root of $T^*$ and $u$ is adjacent to $a$
  and $v$ adjacent to $b$ in $(T',\lambda')$. Thus, $u$ and $v$ are not
  siblings anymore.  Moreover, we keep all edge labelings and set
  $\lambda'(a,u)=\lambda'(b,v)=0$.  By construction,
  $\X[u,v]_{(T',\lambda')}=\{(u,v)\}$.
   We note that $\lca(u,v)$ cannot be the root of $T$ since we have a
  triangle of the form $D$, i.e., there must be an inner 1-edge ancestral
  to $\lca(u,v)$.  One easily checks that $(T',\lambda')$ still explains
  all remaining pairs in $\X\setminus \X[u,v]$.  Hence $(T',\lambda')$
  explains $\X$ whenever $\X[u,v] = \{(u,v)\}$. It is least-resolved by
  construction and Lemma \ref{1-edge}, and thus unique by Thm.\
  \ref{thm:uniqueness}.
		
  Analogously, a tree $(T',\lambda')$ that explains $\X\setminus \X[u,v]$
  with $\X_{T,\lambda}[u,v] = (v,u)$ can be obtained from
  $(T,\lambda)$ by interchanging the roles of $u$ and $v$.
	
  Finally, whenever $(T',\lambda')$ explains either $(u,v)$ or $(v,u)$ we
  can obtain a tree $(T,\lambda)$ that explains $u|v$ by ``reversing'' the
  contraction above. Because of the uniqueness of $(T',\lambda')$ it
  \emph{must} locally look as in Fig.\ \ref{cases}(1a) middle.  That is,
  there is exactly \emph{one} inner 1-edge along the path from $u$ to $v$
  and all edges incident to $\parent(v)$ must be 1-edges.  Hence, after
  collapsing this edge to a single vertex, we obtain the least-resolved
  tree $(T,\lambda)$ that explains $u|v$.  Since
  $\X_{(T',\lambda')}[u,z]=\X_{(T,\lambda)}[u,z]$ and
  $\X_{(T',\lambda')}[v,z]=\X_{(T,\lambda)}[v,z]$ is still true for all
  $z\in L(T')$, i.e., $(T',\lambda')$ explains $\X$ whenever $\X[u,v] =
  u|v$.

  \textbf{Case 1b.} Suppose $\Delta_{uv}$ contains at least one triangle
  each of types $B$, $C$ and $D$, but not of type $A$. Then,
  $\X[u,v]\in\{(u,v),[u,v]\}$.  If $(T,\lambda)$ explains $[u,v]$ it has
  the following properties: Since $\Delta_{uv}$ contains triangles of type
  $B$, $v$ but not $u$ is a sink, and therefore $(\parent(v),v)$ is a
  1-edge while $(\parent(u),u)$ is a 0-edge.  Moreover, since
  $(v,u)\in\X_{T,\lambda}$ and $(\parent(u),u)$ is a 0-edge, the path from
  $\lca(u,v)$ to $u$ has to contain at least one 1-edge, $u$ and $v$ cannot
  have the same parent, hence $\lca(u,v)\succ \parent(u)$. The presence of
  triangles of type $B$, $C$, and $D$ immediately implies that $(u,x)\in\X$
  if and only if $(v,x)\in\X$ for all $x \in L(T)\setminus\{u,v\}$.
  Therefore, since each inner vertex of $(T,\lambda)$ (except possibly the
  root) must be connected to an outer 0-edge (Lemma
  \ref{1-edge}(\ref{it:1}b)), there cannot be any other inner vertex on the
  path from $\lca(u,v)$ to $\parent(u)$, hence the inner edge
  $(lca(u,v),\parent(u))$ must be present in $(T,\lambda)$. On the other
  hand, there must be a 0-edge $(\parent(v),z)$ with $z\in L(T)\setminus
  \{u,v\}$ (Lemma \ref{1-edge}(\ref{it:1}b)).  Moreover, $B$, $C$, and $D$
  imply that there may be other 0- or 1-edges incident to $\parent(v)$.  We
  denote by $T^{**}$ the subtree rooted at $\parent(u)$ that does not
  contain the leaf $u$. The subtree of $T$ that is rooted at $\parent(v)$
  but does neither contain the leaf $v$ nor the leaf $u$ nor any of the
  vertices of $T^{**}$ is denoted by $T^*$.  Thus, $(T,\lambda)$ must match
  the pattern shown Fig.\ \ref{cases}(1b, left).

  A tree $(T',\lambda')$ that explains $\X$ with $\X[u,v]=(u,v)$ can be
  constructed by a simple change in the position of $v$ in $(T,\lambda)$,
  that is, we delete the 1-edge $(\parent(v),v)$ and instead, insert the
  1-edge $(\parent(u),v)$. All other edge labels remain unchanged. By
  construction, $\X[u,v]_{(T',\lambda')}=\{(u,v)\}$ and again, one easily
  checks that $(T',\lambda')$ displays $\X\setminus \X[u,v]$ and therefore
  $\X$. Moreover, $(T',\lambda')$ is by construction least-resolved and
  therefore uniquely defined. Hence, it must locally look as in Fig.\
  \ref{cases}(1b, right). Reverting the local modifications in
  $(T',\lambda')$ again yields the uniquely defined least-resolved tree
  $(T,\lambda)$ that explains $\X$ with $\X[u,v]=[u,v]$.

  \textbf{Case 2a.} Suppose $\Delta_{uv}$ contains at least one triangle
  each of types $A$ and $C$ but no triangles of types $B$ and $D$.  Then
  $\X[u,v]\in\{u|v,(u,v),(v,u)\}$.  We first assume that $(T,\lambda)$
  explains $u|v$. Then, as in Case 1a, $u$ and $v$ have to be siblings,
  none of them is a sink and no other 0-edge is incident to
  $\lca(u,v)$. Hence, $(T,\lambda)$ locally looks again like Case 1a in
  Fig.\ \ref{cases}. Local transformations of $(T,\lambda)$ that are
  completely analogous to Case 1a can be applied to $(T,\lambda)$ in order
  to obtain unique least-resolved trees that explain $\X$ with
  $\X[u,v]=(u,v)$ and $\X[u,v]=(v,u)$, respectively (see Fig.\
  \ref{cases}(1a)). It is not hard to check that these transformations can
  be reversed by contraction of the edge $(v,u)$.

\begin{figure}[t]
  \begin{center}
    \includegraphics[width=\textwidth]{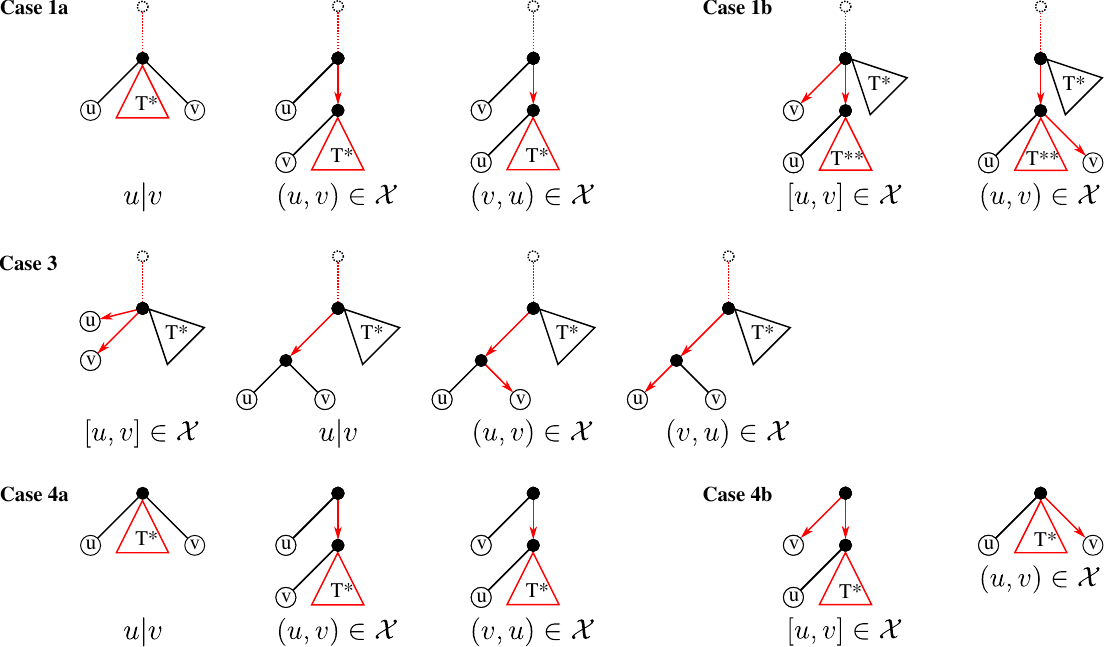}
  \end{center}
  \caption[]{Local modifications of the tree $(T,\lambda)$ necessary to
    explain all possible choices of $\X[u,v]$ for different combination of
    triangles types in $\Delta_{uv}$. Only the local environment around $u$
    and $v$ is shown since the rest of the tree remains unchanged in all
    cases. Dashed lines indicate possible additional subtrees that are
    connected to the local situation by means of a 1-edge. The subtrees
    $T^*$ and $T^{**}$ (with red triangles) may be attached to an inner
    vertex (via 1-edges); their internal structure is irrelevant for the
    arguments in the proof.}
  \label{cases}
\end{figure}

  \textbf{Case 2b.} If $\Delta_{uv}$ contains at least one triangle each of
  types $A$ and $D$ but no triangles of types $B$ and $C$ then exactly the
  same arguments as in Cases 2a and 1a apply.
  
  \textbf{Case 2c.} Suppose $\Delta_{uv}$ contains at least one triangle
  each of types $B$ and $C$ but no triangles of types $A$ or $D$. Then
  $\X[u,v]\in\{(u,v),[u,v]\}$. Let us assume that $(T,\lambda)$ explains
  $[u,v]\in\X$. As in Case 1b, the presence of triangles of type $B$
  implies that $v$ but not $u$ is a sink. Arguing as in Case 1b shows that
  $(T,\lambda)$ locally looks like Case 1b in Fig.\ \ref{cases}. The local
  transformation to the least-resolved tree $(T',\lambda')$ that explains
  $\X$ with $\X[u,v]=(u,v)$ can be performed as described in Case 1b,
  resulting in a tree $(T',\lambda')$ that locally looks like Fig.\
  \ref{cases}(1b).  The same arguments as in Case 1b can be applied to show
  that $(T',\lambda')$ explains $\X$ and that there is a uniquely defined
  reverse transformation that converts $(T',\lambda')$ into $(T,\lambda)$.

  \textbf{Case 2d.} If $\Delta_{uv}$ contains at least one triangle each of
  types $B$ and $D$ but no triangles of types $A$ and $C$ exactly the same
  arguments as in Cases 2c and 1b apply.

  \textbf{Case 3.} Suppose $\Delta_{uv}$ contains at least one triangle
  each of types $C$ and $D$ but no triangles of types $A$ and $B$. Then,
  $\X[u,v] \in \{[u,v],(u,v),(v,u), u|v\}$. Assume that $(T,\lambda)$
  explains $[u,v]\in\X$. This implies that both $u$ and $v$ are sinks of
  $\X$, i.e., $(\parent(u),u)$ and $(\parent(v),v)$ are both 1-edges. By
  symmetry, $(u,x)\in\X$ if and only if $(v,x)\in\X$ and $(x,u)\in\X$ if
  and only if $(x,v)\in\X$, respectively, holds for all $x\in
  L(T)\setminus\{u,v\}$.

  We continue to show that $u$ and $v$ must be siblings.  Assume for
  contradiction, they are not.  Lemma \ref{1-edge}(\ref{it:1}) implies that
  there are leaves $z,z'\in L(T)\setminus\{u,v\}$ such $(\parent(u),z)$ and
  $(\parent(v),z')$ are 0-edges.  Hence, we have at least one of the cases
  $(u,z)\notin\X$ but $(v,z)\in\X$ or $(v,z')\notin\X$ but $(u,z')\in\X$.
   If $\parent(u)$ and $\parent(v)$ are incomparable in $T$, even both cases
  are true. However, we obtain a contradiction to ``$(v,x)\in\X$ iff
  $(u,x)\in\X$''. Thus, $u$ and $v$ are siblings.  Since $(\parent(u),u)$
  and $(\parent(v),v)$ are both 1-edges, there must be a leaf $y\in
  L(T)\setminus\{u,v\}$ such that the edge $(\lca(u,v),y)$ is a 0-edge by
  Lemma \ref{1-edge}(\ref{it:1}).  We denote by $T^*$ the subtree rooted at
  $\lca(u,v)$ without the leaves $u$ and $v$.

  Given $(T,\lambda)$, a tree $(T',\lambda')$ that displays $\X$ with
  $\X[u,v]=u|v$ is obtained by inserting an inner 1-edge $(a,b)$ such that
  $a$ becomes the new root of $T^*$ and $b=\lca_{T'}(u,v)$. The outer edges
  $(b,u)$ and $(a,v)$ are 0-edges; all other edge labels are retained as in
  $(T,\lambda)$. The resulting tree locally looks as illustrated in Fig.\
  \ref{cases}(Case 3). Relabeling of edges in $(T',\lambda')$ such that
  $(b,v)$ becomes a 1-edge yields the tree $(T'',\lambda'')$ that
  explains $\X$ with $(u,v)\in\X$. Similarly, converting the edge $(b,u)$ of
  $(T',\lambda')$ into a 1-edge yields the tree $(T''',\lambda''')$ that
  explains $\X$ with $(v,u)\in\X$.  The trees $(T',\lambda')$,
  $(T'',\lambda'')$, and $(T''',\lambda''')$ explain $\X$ with the
  corresponding choice of $\X[u,v]$ and are least-resolved and thus unique.
  As in the previous cases, the reverse transformations are therefore also
  uniquely defined.
  
  \textbf{Case 4a.} Suppose that all triangles in $\Delta_{uv}$ are of the
  form $A$. Then $\X[u,v]\in\{u|v,(u,v),(v,u)\}$. Let us assume that
  $(T,\lambda)$ displays $u|v$. Then, $u$ and $v$ are both sources, hence
  $(\parent(u),u)$ and $(\parent(v),v)$ are both 0-edges.  Note that in
  contrast to Case 1a, there is no $x\in L(T)\setminus\{u,v\}$ with
  $(x,u)\in\X$ or $(x,v)\in\X$.  This implies that $u$ and $v$ are both
  incident to the root $\rho_T$ of $(T,\lambda)$ and among all edges
  incident to the root, $(\rho_T,u)$ and $(\rho_T,v)$ are the only
  0-edges. The tree $(T,\lambda)$ explaining $\X[u,v]=u|v$ is shown Fig.\
  \ref{cases}(Case 4a). Note, the tree structure is very similar to Case
  1a. Therefore, as in Case 1a, $(T,\lambda)$ can be locally modified to a
  least-resolved tree $(T',\lambda')$ explaining $\X$ with $\X[u,v]=(u,v)$
  by introducing the single 1-edge $(a,b)$ with $a=\parent(u),
  b=\parent(v)$. The vertex $b$ becomes the root of $T^*$, where $T^*$ is
  defined as in Case 1a (see Fig.\ \ref{cases}(Case 4a)). We set
  $\lambda'(a,u)=0$ and $\lambda'(b,v)=0$, while all other edge-labels are
  retained.

  Exchanging the roles of $u$ and $v$ in $(T',\lambda')$ defines a
  least-resolved tree $(T'',\lambda'')$ that explains $\X[u,v]=(v,u)$. As
  in the previous cases, one easily verifies that all resulting trees are
  least-resolved and explain $\X$ with the corresponding choice for
  $\X[u,v]$. Hence the reverse transformations are also uniquely defined.

  \textbf{Case 4b.}  Suppose that $\Delta_{uv}$ contains only triangles of
  the form $B$.  Hence, $\X[u,v]\in\{(u,v),[u,v]\}$. Let us first assume
  that the least-resolved tree $(T,\lambda)$ explains $[u,v]\in\X$. It
  immediately follows that $v$ is a sink and $u$ is not, hence
  $\lambda(\parent(v),v)=1$ and $\lambda(\parent(u),u)=0$. Moreover, we
  have $(v,u)\in\X$, thus $\parent(v)\succ\parent(u)$. Since for any $x\in
  L\setminus\{u,v\}$ it holds $(x,u)\notin\X$ and thus, $\parent(u)\succeq
  \lca(u,x)$, we have $\parent(v)=\rho_T$ and $\deg(\rho_T)=2$. Therefore,
  $(T,\lambda)$ locally looks as in Fig.\ \ref{cases}(Case 4b).  Note that
  the tree structure is very similar to Case 1b.  Hence, similar as in Case
  1b, $(T,\lambda)$ can be modified locally to a least-resolved tree
  $(T',\lambda')$ that displays $(u,v)\in\X$ by contraction of
  $(\parent(v),\parent(u))$ and keeping all other edge-labels (see Fig.\
  \ref{cases}(Case 4b, right)). By the same argumentation as before, the
  reverse transformation is also uniquely defined.

  \textbf{Case 4c.} Let us assume that all triangles in $\Delta_{uv}$ are
  of the form $C$, i.e., $\X[u,v]\in\{[u,v],u|v,(u,v),(v,u)\}$, and that
  $(T,\lambda)$ explains $[u,v]\in\X$. As in Case 3, both $u$ and $v$ are
  sinks of $\X$, i.e., $(\parent(u),u)$ and $(\parent(v),v)$ are both
  1-edges. Using the same symmetry argument as in Case 3, we conclude for
  any $x\in L(T)\setminus\{u,v\}$ that $(u,x)\in\X$ if and only if
  $(v,x)\in\X$, and $(x,u)\in\X$ if and only if $(x,v)\in\X$,
  respectively. Following the arguments laid out in Case 3, we conclude
  that $(T,\lambda)$ locally looks as Case 3 of Fig.\ \ref{cases}. Thus the
  local transformations described above can be applied analogously in order
  to obtain least-resolved trees that explain all possible $\X[u,v]$.

  \textbf{Case 4d.} If $\Delta_{uv}$ contains only triangles of the form
  $D$, then we can apply the same construction as in Case 4a and 3 in order
  to conclude that $\X$ can be explained for all possible $\X[u,v]$.

 \end{proof}

\section{Algorithmic Considerations}

\NEW{Summarizing our results, we present two different algorithms that
  are both able to recognize a Fitch relation and compute its unique
  least-resolved tree. The first algorithm checks all induced triangles for
  forbidden subgraphs and, once recognized a Fitch relation, uses the set
  of informative triple as an input for the algorithm \texttt{BUILD}. Then,
  it simply labels the edges of the resulting Aho tree in the correct way.
  This is a very intuitive way to check for Fitch relations and construct
  the least-resolved tree, which we will make precise first.  We shall see
  that it is possible, however, to achieve a much better performance by
  using that fact that Fitch graphs are di-cographs. One can alternatively
  check for Fitch relations using properties of di-cographs and build the
  least-resolved tree from the corresponding cotree. This can be achieved
  in linear time.}

We have seen in the previous sections that every valid relations $\X$ is
explained by a unique, least-resolved tree $(T_{\X},\lambda_{\X})$, which,
in turn, is identified by a set $\IT{\X}$ of informative triples due to
Lemma~\ref{lem:cl}. Lemma~\ref{g1} therefore implies
\begin{equation}
  T_{\X} = \Aho(\IT{\X})
\end{equation}
It remains to construct the labeling function $\lambda_{\X}$ on
$\Aho(\IT{\X})$.

\begin{algorithm}[tbp]
\caption{\texttt{Label the Aho tree}}
\label{alg:labelit}
\begin{algorithmic}[1]
  \Require $T_{\X}=\Aho(\IT{\X})$; 
  \Ensure Least-resolved edge-labeled tree $(T_{\X},\lambda_{\X})$ for $\X$; 
  \For{all $e=(u,v)\in E(T)$}
     \If{$v\notin L$} 
        $\lambda_{\X}(e)=1$; 
     \Else 
        \If {$(x,v)\in\X$ for all $x\in L\setminus\{v\}$} 
           $\lambda_{\X}(e)=1$; 
        \Else\ 
           $\lambda_{\X}(e)=0$;
        \EndIf
     \EndIf
  \EndFor
\end{algorithmic}
\end{algorithm}

\begin{lemma} 
  Given the topology $T_{\X}$ of the unique least-resolved tree explaining
  $\X$, Algorithm~\ref{alg:labelit} computes its correct unique edge
  labeling $\lambda_{\X}$ in $\mathcal{O}(\max\{|\X|,|L|\})$ time. 
	\label{lem:alg1}
\end{lemma}
\begin{proof} 
  By Lemma~\ref{1-edge} all inner edges $e$ of $\Aho(\IT{\X})$ must be
  labeled $\lambda(e)=1$ since otherwise they could be contracted, and
  hence, the tree would not be least-resolved. Now consider an edge
  $e=(u,v)$ leading to a leaf $v\in L$. If $(x,v)\notin\X$ for some $x\in
  L\setminus\{v\}$ then $\lambda(e)=0$. Conversely, if $\lambda(e)=0$ then
  $(x,v)\notin\X$ for every leaf below the siblings of $u$. At least one
  such leaf $x$ exists in a phylogenetic tree. Hence an outer edge is
  labeled $\lambda(e)=1$ if and only if $(x,v)\in\X$ for all $x\in
  L\setminus\{v\}$.

  For the time-complexity note that the labeling
  Algorithm~\ref{alg:labelit} requires $\mathcal{O}(|E(T_{\X}|))$
  operations to label the inner edges. To label the $|L|$ outer edges
  $(u,v)$ we have to determine the degree of vertex $v$ in $\X$, that is,
  $\deg(v)=1$ implies that $(u,v)$ is an outer edge, which requires
  $\mathcal{O}(\max\{|\X|,|L|\})$ operations. Since $|E(T_{\X})|$ is bounded by
  $\mathcal{O}(|L|)$, the total running time of the labeling step is
  bounded by $\mathcal{O}(\max\{|\X|,|L|\})$.   
\end{proof}

\NEW{A tree explaining a given Fitch relation can be obtained by the
  following procedure: First, we check whether $\X$ is a Fitch
  relation. This} can be achieved in $\mathcal{O}(|L|^3)$ by checking
validity of the $\binom{L}{3}$ induced triangles.  If $\X\subset L\times L$
is a Fitch relation, then $\IT{\X}$ can be constructed within
$\mathcal{O}(|L|^3)$ time.  For a given the set of triples $R=\IT{\X}$, the
original approach to check whether $R$ is consistent (in which case
$\Aho(R)$ is returned) or not, has time complexity $\mathcal{O}(|R||L|)$
\cite{aho_inferring_1981}. However, various further practical
implementations have been described
\cite{henzinger_constructing_1999,Jansson:05,Holm:01,DF:16} \NEW{that
  improve the asymptotic performance. Constructing $\Aho(R)$ and using
  Algorithm \ref{alg:labelit} to obtain the edge labels, it is therefore
  possible to recognize a Fitch relation $\X$ and to compute its respective
  (least-resolved) tree $(T,\lambda)$ in $\mathcal{O}(|L|^4)$. }

\NEW{It is possible to improve the algorithms to recognize Fitch relations
  $\X$ and compute its least-resolved tree $(T,\lambda)$ in the following
  way:} Every di-cograph $G$ is explained by a unique cotree $(T',t)$
\cite{Moehring:84,McConnell:05}, that is, an ordered phylogenetic tree $T'$
with leaf set $V(G)$ and a vertex-labeling function $t:V^0(T')\to
\{0,1,\overrightarrow{1}\}$\NEW{, such that $t(u)\neq t(v)$ for all inner
  edges $(u,v)$ in $T'$,} defined by
\begin{align*}
  t(\lca(x,y)) = \begin{cases}
    0,    & \text{ if } (x,y)(y,x)\notin E(G) \\
    1,    & \text{ if } (x,y)(y,x)\in E(G)   \\
    \overrightarrow{1}, &\text{ else } .\,
 \end{cases}
\end{align*}
Since the vertices in the cotree $T'$ are ordered, the label
$\overrightarrow{1}$ on some $\lca(x,y)$ of two distinct leaves $x,y\in L$
means that there is an edge $(x,y) \in E(G)$, while $(y,x) \notin E(G)$,
whenever $x$ is placed to the left of $y$ in $T'$ \cite{Hellmuth:16a}.  
As discussed in Section \ref{sec:fitch-graph}, any di-cograph that does not
contain the invalid triangles $F_1$, $F_5$ and $F_8$ is a Fitch graph.
\begin{lemma}
  Let $G$ be a di-cograph and $(T',t)$ its corresponding cotree.  A
  di-cograph contains the triangle $F_1$, $F_5$ and $F_8$ as an induced
  subgraph if and only if there are two vertices $v,w\in V^0(T')$ with
  $v\succ_{T'} w$ such that either (i) $t(v)=0\neq t(w)$ or (ii)
  $t(v)=\overrightarrow{1}$, $t(w)= {1}$ and \NEW{$w$ is located in some
    subtree (rooted at a child of $v$) that is different from the subtree
    rooted at the right-most child of $v$.  }
  \label{lem:cotree-triangle}
\end{lemma}
\begin{proof}
  Consider first the triangles $F_1$ and $F_5$ with vertices $x,y,z$ and
  edge set $E(F_1)=\{(x,y)\}$ and $E(F_5)=\{(x,y), (y,x)\}$.  Equivalently,
  we have $t(\lca_{T'}(x,y))\in \{1,\overrightarrow{1}\}$,
  $t(\lca_{T'}(x,y,z)) =0$ and $v=\lca_{T'}(x,y,z) \succ_{T'}
  w=\lca_{T'}(x,y)$.

  Now, let $F_8$ have vertices $x,y,z$ and edge set $E(F_8) = \{(x,y),
  (y,x), (x,z), (y,z)\}$.  Equivalently, we have $t(\lca_{T'}(x,y)) = 1$,
  $t(\lca_{T'}(x,y,z)) =\overrightarrow{1}$ and $v=\lca_{T'}(x,y,z) \succ_T
  w=\lca_{T'}(x,y)$. \NEW{In particular, $x$ and $y$ must be placed left
    from $z$ in ${T'}$ and therefore, $w$ must be located in some subtree
    different from the subtree rooted at the right-most child of
    $v=\lca_T(x,y,z)$.}   
\end{proof}

\begin{corollary}
  Let $\X$ be a Fitch graph and $(T',t)$ its corresponding cotree.  If $\X$
  contains an edge, then it is weakly connected\NEW{, i.e., the underlying
    undirected graph obtained from $\X$ by ignoring the direction of the
    edges is connected. Moreover, any vertex $x\prec v$ for which $t(v)=0$
    must be a leaf of $T'$.}
\label{cor:0-leaf}
\end{corollary}
\begin{proof}
  If a Fitch graph $\X$ contains an edge, then its cotree contains an inner
  vertex labeled $1$ or $\overrightarrow{1}$.  If $\X$ is disconnected,
  then the root of the cotree must be labeled $0$ and Lemma
  \ref{lem:cotree-triangle} implies that $\X$ is not a Fitch graph.  Thus, the
  root must be labeled either $1$ or $\overrightarrow{1}$, which implies
  that $\X$ is weakly connected.  

  \NEW{Now assume that $(T',t)$ contains a vertex $v$ with $t(v)=0$.  Let
    $x\prec v$ with $(v,x)\in E(T')$ and assume for contradiction that $x$
    is an inner vertex.  By the definition of cotrees, $t(v)=0\neq t(x)$.
    Lemma \ref{lem:cotree-triangle} and Theorem \ref{thm:main} imply that
    $\X$ is not a Fitch graph; a contradiction.}   
\end{proof}

Verifying whether a graph $G$ is a di-cograph or not can be achieved
\NEW{in} $\mathcal{O}(|V(G)|+|E(G)|)$ time, see \cite{McConnell:05,
  Hellmuth:16a} for further details.  To verify that a given di-cograph $G$
does not contain $F_1$, $F_5$ and $F_8$ as an induced subgraph, we apply
\NEW{the classical} Breadth-first search (BFS)
\cite{cormen2009introduction} on its cotree $(T',t)$ starting with the root
and check whether there are invalid combinations of vertex labels in
$(T',t)$ according to Lemma \ref{lem:cotree-triangle}.  Note, $L(T')=V(G)$
and $|V^0(T')|\leq |L(T')|-1$.  Thus, the BFS-method runs in
$\mathcal{O}(|V(T')|) = \mathcal{O}(|V(G)|)$ time.  Therefore, recognition
of Fitch graphs, or equivalently, Fitch relations can be achieved within
$\mathcal{O}(|V(G)|+|E(G)|)$ time.

\NEW{We show now how to obtain a tree $(T,\lambda)$ that explains a Fitch
  relation $\X$ from its cotree representation $(T',t)$. To this end we
  need to translate the (ordered) cotree with vertex labels ``$0$'',
  ``$1$'' and ``$\overrightarrow{1}$'' to an unordered tree with edge
  labels ``$1$'' and ``$0$'', summarized next and called
  \texttt{cotree2fitchtree}:
  \begin{description}
  \item[For all $x\in V^0(T')$, if ] \ \smallskip
  \item[]$t(x)=1$ (resp.\ $0$), then
    set for each child $y$ of $x$ the label $\lambda(x,y)=1$ 
    (resp.\ $0$), and else,
  \item[]$t(x)=\overrightarrow{1}$, then we can assume w.l.o.g.\ that the
    children of $x$ are ordered $x_1,\dots,x_k$, $k\geq 2$ from left to
    right. Now, replace the subtree of $T'$ with vertices $x$ and
    $x_1,\dots,x_k$ by the caterpillar $C(x_1,\dots,x_k) \coloneqq
    (x_1(x_2(\dots (x_{k-1},x_k)\ldots)$ \emph{(in Newick notation)} that is
    rooted at $x$.  Set the label $\lambda$ of all inner edges of
    $C(x_1,\dots,x_k)$ and the outer-edge incident to $x_k$ to ``$1$'' and
    the labels of all other (outer) edges of $C(x_1,\dots,x_k)$ to
    ``$0$''. Note that outer edges of $C(x_1,\dots,x_k)$ may be inner edges
    in $(T,\lambda)$. \smallskip
  \item[]Finally, remove all vertex labels and ignore the ordering of
    the vertices to obtain the tree $(T,\lambda)$.
  \end{description}

  For an example of \texttt{cotree2fitchtree} see Figure
  \ref{fig:cotree2fitchtree}.
	
  \begin{figure}[tbp]
    \begin{center}
      \includegraphics[width=\textwidth]{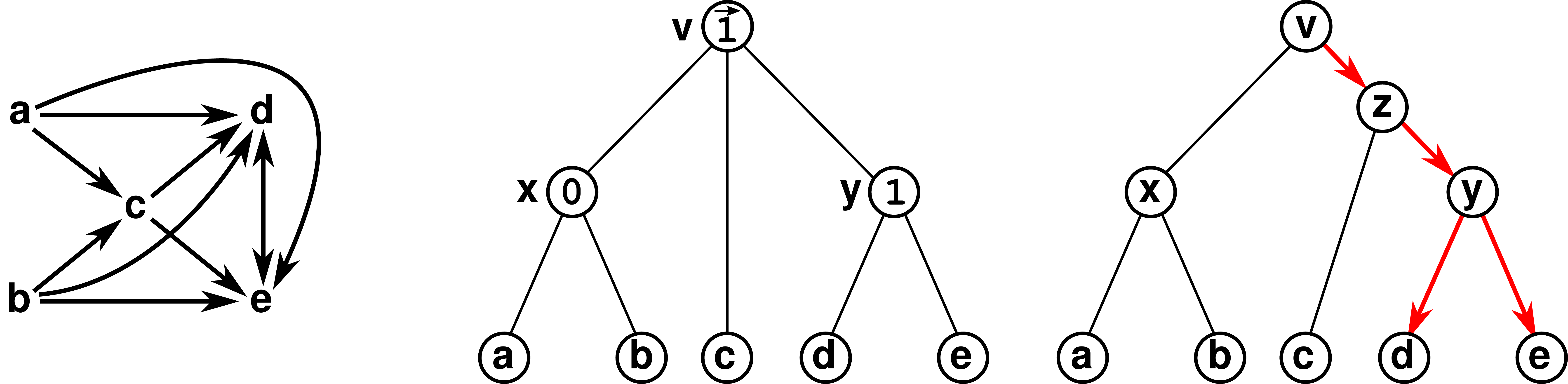}
    \end{center}
    \caption[]{Application of \texttt{cotree2fitchtree}: A Fitch-relation
      $\X$ (left), its cotree $(T',t)$ (middle) and a tree $(T,\lambda)$
      that explains $\X$ (right) is shown. The tree $(T,\lambda)$ is
      obtained from $(T',t)$ by replacing the subtree with vertices $v,x,y$
      and $c$ by the caterpillar $(x(c,y))$ rooted at $v$ and adding the
      edge-labels as described in the procedure \texttt{cotree2fitchtree}.
      By Lemma \ref{lem:cotree-triangle}, $t(x)=0$ for all inner vertices
      $x$ in the subtrees left from the subtree rooted a the right-most
      child $y$ of $v$.  Note, $(T,\lambda)$ is not least-resolved w.r.t.\
      $\X$. Nevertheless, Theorem \ref{thm:uniqueness} implies that
      $(T,\lambda)$ is displayed by the least-resolved tree for $\X$. Here,
      the least-resolved tree can be obtained from $(T,\lambda)$ by
      contracting the edges $(v,x)$ and $(z,y)$. }
  \label{fig:cotree2fitchtree}
\end{figure}

\begin{lemma}
  The procedure \texttt{cotree2fitchtree} transforms the cotree $(T',t)$ of
  a Fitch-relation $\X$ into a tree $(T,\lambda)$ that explains $\X$ in
  $\mathcal{O}(|V(T')|)$ time.
\end{lemma}
\begin{proof}
  Let $(T',t)$ be the cotree of the Fitch-relation $\X$ and $(T,\lambda)$
  the tree resulting from \texttt{cotree2fitchtree}.  Since all inner
  vertices of $(T',t)$ are labeled, each edge of $(T,\lambda)$ receives a
  label ``$0$'' or ``$1$'' by construction.  It needs to be verified that
  $(T,\lambda)$ explains $\X$.
		
  Assume that $(x,y),(y,x)\in \X$. Hence, $t(\lca_{T'}(x,y)) =1$.  By
  construction, the edges incident to the children of $v=\lca_{T'}(x,y)$
  are labeled ``$1$''. Hence, both paths in $(T,\lambda)$ from
  $\lca_{T}(x,y)=v$ to $x$ and to $y$ contain 1-edges.  Thus, $(T,\lambda)$
  explains all symmetric pairs in $\X$.

  Assume that $(x,y),(y,x)\notin \X$ and let $z=\lca_{T'}(x,y)$.  Hence,
  $t(z) =0$. Cor.\ \ref{cor:0-leaf} implies that $(z,x)$ and $(z,y)$ are
  outer edges in $T'$ that are, by construction, labeled ``$0$'' in
  $(T,\lambda)$.  As a consequence, the path from $x$ to $y$ in
  $(T,\lambda)$ contains only 0-edges, which implies that $(T,\lambda)$
  also explains that all pairs $(x,y),(y,x)$ that are not contained in
  $\X$.
	
  Assume $(x,y)\in \X$ and $(y,x)\notin \X$.  Hence, $t(\lca_{T'}(x,y))
  =\overrightarrow{1}$ and $x$ is left from $y$ in $T'$.  Let $v_i$ and
  $v_j$ be children of $\lca_{T'}(x,y)$ with $v_i\succeq x$ and $v_j\succeq
  y$.  Since $x$ is left from $y$, also $v_i$ is left from $v_j$ in $T'$.
  Note, $v_i$ and $v_j$ are now part of the inserted caterpillar
  $C(\child(\lca_{T'}(x,y)))$ in $(T,\lambda)$.  Therefore, $\lca_T(x,y)$
  must be an inner vertex of this caterpillar.  By construction, the path
  from $\lca_T(x,y)$ to $v_j\succeq y$ contains a 1-edge and thus $(x,y)\in
  \X$.  It remains to show that the path from $\lca_T(x,y)$ to $x$ contains
  only 0-edges so that $(y,x)\notin \X$.  Note that the vertex $v_i$ is a
  child of $\lca_T(x,y)$ in $T$ and the edge $(\lca_T(x,y),v_i)$ is labeled
  ``$0$''.  Thus, if $v_i=x$ we are done. Assume that $v_i\neq x$ and
  hence, that $v_i$ is an inner vertex of $T'$.  By the definition of
  cotrees, we have $t(\lca_{T'}(x,y))=\overrightarrow{1}\neq t(v_i)$.
  Since $v_i$ is left from $v_j$ in $T'$ we can apply Lemma
  \ref{lem:cotree-triangle} and conclude that $t(v_i)\neq 1$. Hence, there
  is only one possibility left, namely $t(v_i)=0$.  Cor.\ \ref{cor:0-leaf}
  implies that $(v_i,x)$ must must be an outer edge in $(T',t)$ that -- by
  construction -- is labeled ``$0$'' in $(T,\lambda)$. Hence, the path from
  $\lca_{T}(x,y)$ to $x$ contains only 0-edges and therefore, $(y,x)\notin
  \X$.

  For the running time, observe that the edge-label in each step of
  \texttt{cotree2fitchtree} for vertices $v$ with $t(v)\in \{0,1\}$ can be
  computed in $\mathcal{O}(\deg_{T'}(v))$ time.  Moreover, if
  $t(v)=\overrightarrow{1}$ for some vertex $v$ in $(T',t)$, we have to
  replace the subtree induced by $v$ and its children $v_1,\dots,v_k$
  (ordered from left to right) in $(T',t)$, by the edge-labeled caterpillar
  $C(v_1,\dots,v_k)$.  This task can also be performed in
  $\mathcal{O}(\deg_{T'}(v))$ time.  Since each step in
  \texttt{cotree2fitchtree} can be done in $\mathcal{O}(\deg_{T'}(v))$ time
  and $\sum_{v\in V^0(T')} \deg_{T'}(v)) \leq 2|E(T')| < 2|V(T')|$, this
  implies a total time requirement of $\mathcal{O}(|V(T')|)$.   
\end{proof}

Let $(T,\lambda)$ be the tree that explains $\X$ as constructed with
\texttt{cotree2fitchtree} from the respective cotree $(T',t)$.  Theorem
\ref{thm:uniqueness} implies that $(T,\lambda)$ displays the least-resolved
tree for $\X$. Thus, we can utilize Lemma \ref{1-edge} and contract all
irrelevant edges and all inner 0-edges in $(T,\lambda)$ in order to obtain
the least-resolved tree for $\X$.  The latter can be done in
$\mathcal{O}(|V(T)|)$ time.  Taking the latter results together with the
observation that $|V(T)|\geq |V(T')|$, we obtain the following
\begin{theorem}
  Verifying whether an irreflexive relation $\X\subseteq L\times L$ is a
  Fitch relation or not, can be a achieved in $\mathcal{O}(|L|+|\X|)$
  time. Its unique least-resolved edge-labeled tree $(T_{\X},\lambda_{\X})$
  can be computed in $\mathcal{O}(|V(T_{\X})|) = \mathcal{O}(|L|)$ time,
  given the cotree of $\X$.
\end{theorem}
}

The fact the Fitch graphs form a heritable family (cf.\ Lemma~\ref{lem:Fitch})
has far-reaching consequences for computational problems such as the
following:

\NEW{
\begin{problem}[Fitch Graph Modification] \ \\ 
\label{prob:modi}
\begin{tabular}{ll}
  \emph{Given:}    & a graph $G=(L,F)$ and integers $i,j,k$. \\
  \emph{Question:} & Are there subsets $L'\subseteq L$, $F'\subseteq F$
  and $F''\subseteq (L\times L)\setminus F$ \\
		   & with $|L'|\leq i$, $|F'|\leq j$ and $|F''|\leq k$ 
                     such that \\
		   & $G-L'-F'+F''$ is a Fitch graph?
\end{tabular}
\end{problem} 
}

\NEW{A very general result on graph editing on heritary graph classes
  \cite{GRRW:10,CAI:1996} immediately implies}
\begin{corollary}
  Fitch Graph Modification is NP-complete, but fixed-parameter tractable
  and can be solved in \NEW{$\mathcal{O}(3^{i+2j+2k}|L|^4)$} time.
\end{corollary}

\section{Concluding Remarks}

The relationships and mutual constraints of gene trees and species trees
are by no means completely understood. Here we have attempted to identify
the phylogenetic information that is contained in horizontal transfer
events. An alternative approach to understand such relations has been
explored by \cite{Hellmuth:16a}. The relations considered there, however,
are completely defined by the labeling of the inner gene tree vertices as
speciation, duplication or HGT.

A more commonly used definition of xenology was proposed by Walter Fitch
\cite{Fitch:00}.  We formalized Fitch's concept of xenology in the form of
a not necessarily symmetric binary relation $\mathcal{X}$ so that
$(x,y)\in\mathcal{X}$ if and only the lineage from $\lca(x,y)$ to $y$ was
horizontally transferred at least once. Our main result is a complete
characterization of such relations in terms of forbidden induced subgraphs
and a complete characterization of the minimally resolved trees explaining
such relations. These Fitch trees represent the complete information on the
gene tree that is ``recorded'' by the horizontal transfer events
alone. Polynomial-time algorithms have been devised to compute Fitch trees
from Fitch relations.

The practical usefulness of the Fitch relation and its trees eventually
will depend on how easy or difficult it will turn out to estimate the Fitch
relation from data. Although no convenient tools are available to our
knowledge to identify directed xenology relationships without first
reconstructing gene and species trees, this seems to be not at all a
hopeless task, since genes that are imported by HGT from an ancestor of
species $A$ into an ancestor of species $B$ are expected to be more closely
related than expected from the bulk of the genome
\cite{Novichkov:04,Ravenhall:15}. Inference from real-life data will never
be noise free. It is therefore encouraging that the corresponding editing
problem is at least FPT even though it is NP complete as so many other
computational problems in phylogenetics.

An interesting facet of the results is that the Fitch graphs are a proper
subset of the di-cographs that naturally appear in a formalization of
xenology that focuses on the vertices of the gene tree
\cite{Hellmuth:16a}. While our results strongly suggest that there should
be a close relationship between these two models, it remains an open
question what exactly this connection and its biological interpretation
might be. A related question concerns the symmetrized version of the
xenology relation: what can be said about the relation
$\mathcal{X}^{\textrm{sym}}$ with $\{x,y\}\in\mathcal{X}^{\textrm{sym}}$
whenever $(x,y)\in\mathcal{X}$, that is,
$\{x,y\}\in\mathcal{X}^{\textrm{sym}}$ iff there is a HGT event along the
unique path from $x$ to $y$ in the gene tree? What if we knew that there is
exactly one transfer event along the path?

\section*{Acknowledgements}
  We thank Maribel Hern{\'a}ndez Rosales and her team for stimulating
  discussions. This work was funded in part by 
	the BMBF-funded project ``Center for RNA-Bioinformatics'' 
	(031A538A, de.NBI-RBC) 
 	and a travel
  grant from DAAD PROALMEX (Proj.\ No.\ 278966).

\bibliographystyle{abbrv}
\bibliography{literature}

\end{document}